\documentclass[authoryear,a4paper,review]{elsarticle}
\pdfoutput=1 

\journal{European Journal of Operational Research}

\usepackage{amsmath,amssymb,amsthm,graphicx,hyperref,setspace,subfigure,float,multirow} 
\usepackage[authoryear]{natbib}
\usepackage[boxruled]{algorithm2e}
\usepackage[left=1in,right=1in,top=1.4in,bottom=1.4in]{geometry}
\hypersetup{
	pdftitle={},
	pdfauthor={Jaehyuk Choi},
	pdfkeywords={Stochastic volatility, SABR model, CEV model},
	colorlinks=true,
	linkcolor=red,
	citecolor=blue,
	urlcolor=blue,
	bookmarksnumbered=true,
	pdfstartview=
}

\date{September 19, 2025}

\newcommand{\rhoc}{\rho_\ast}
\newcommand{\betac}{\beta_\ast}

\newcommand{\vov}{\nu}

\newcommand{\dt}{h}
\newcommand{\vovn}{\hat{\vov}}
\newcommand{\IV}{I_0^T}
\newcommand{\IVstep}{I_t^h}
\newcommand{\condF}{\bar{F}_0^T}
\newcommand{\condFstep}[1][]{\bar{F#1}_t^{\dt}}

\newcommand{\cev}[1][]{\mathcal{CEV}_{\beta#1}}
\newcommand{\GAM}{\mathcal{G}}
\newcommand{\NCX}{{\chi^2}}
\newcommand{\POIS}{\mathcal{POIS}}
\newcommand{\SP}{\mathcal{SP}}

\newcommand{\acosh}{\operatorname{arcosh}}

\newcommand{\qtext}[2][\quad]{#1\text{#2}#1}
\newtheorem{proposition}{Proposition}
\newtheorem{definition}{Definition}
\newtheorem{remark}{Remark}

\begin{document}

\begin{frontmatter}
\title{Efficient and accurate simulation of the stochastic-alpha-beta-rho model}

\author[cu]{Jaehyuk Choi\corref{corrauthor}}
\ead{jc6569@columbia.edu}

\cortext[corrauthor]{Correspondence. \textit{Tel:} +1-212-854-2643, \textit{Address:} Rm 426 Mathematics Hall, 2990 Broadway, New York, NY 10027}

\address[cu]{Department of Mathematics, Columbia University, United States}
\author[b]{Lilian Hu}
\author[b]{Yue Kuen Kwok}
\address[b]{Financial Technology Thrust, Hong Kong University of Science and Technology (Guangzhou), China}

\begin{abstract}
We propose an efficient, accurate and reliable simulation scheme for the stochastic-alpha-beta-rho (SABR) model. The two challenges of the SABR simulation lie in sampling (i) integrated variance conditional on terminal volatility and (ii) terminal forward price conditional on terminal volatility and integrated variance. 
For the first sampling procedure, we sample the conditional integrated variance using the moment-matched shifted lognormal approximation. For the second sampling procedure, we approximate the conditional terminal forward price as a constant-elasticity-of-variance (CEV) distribution. Our CEV approximation preserves the martingale condition and precludes arbitrage, which is a key advantage over Islah's approximation used in most SABR simulation schemes in the literature. We then adopt the exact sampling method of the CEV distribution based on the shifted-Poisson mixture Gamma random variable. Our enhanced procedures avoid the tedious Laplace inversion algorithm for sampling integrated variance and non-efficient inverse transform sampling of the forward price in some of the earlier simulation schemes. Numerical results demonstrate our simulation scheme to be highly efficient, accurate, and reliable.
\end{abstract}

\end{frontmatter}

\section{Introduction} \noindent
The stochastic-alpha-beta-rho (SABR) model proposed by \cite{hagan2002sabr} has been widely adopted in option pricing due to its ability to capture volatility smile or skew using a few parameters, and exhibiting consistency in revealing the dynamic behavior between asset price and volatility smile. On one hand, it is common in equity options that the implied volatility of out-of-money options is generally higher than that of the  in-the-money counterparts, known as the volatility skew. On the other hand, the relatively symmetric volatility smile is common in the foreign exchange options. The SABR model can capture the volatility skew and volatility smile in the respective markets. In addition, the SABR model exhibits consistency in the dynamic behavior of volatility smile or skew, where an increase of the price of the underlying asset leads to shifting of the volatility smile or skew to a higher price \citep{hagan2002sabr}. The SABR model can generate such co-movements correctly while the local volatility model fails to do so \citep{derman1998stochastic,dupire1997pricing}.

In their original paper, \cite{hagan2002sabr} derived an analytic implied volatility formula under the SABR model via an asymptotic expansion in small time-to-maturity. It has been a standard practice for practitioners to obtain the option price from the implied volatility using the Black--Scholes formula. However, \citet{hagan2002sabr}'s approximation becomes unreliable under large time-to-maturity or when the option is deep-out-of-the-money. As a result, one important stream of the SABR model research has been improving the implied volatility approximation~\citep{obloj2007fine,paulot2015asym,lorig2017lsv,gulisashvili2018mass,yang2017cev-sabr,choi2021sabrcev}. Notably, \cite{antonov2012adv} obtained a more accurate approximation by mapping the implied volatilities under the correlated cases to those of the uncorrelated cases, under which more analytic properties of the SABR model become available.

This paper follows the other stream of research, which is the development of efficient simulation algorithms for pricing options under the SABR model. With growing popularity of the SABR model, there has been demand for pricing path-dependent derivatives using the model. Since analytic approximation approaches are limited to pricing European options, the Monte Carlo method becomes a better choice.
In particular, the research on the SABR simulation started with the less than perfect performance of the Euler and Milstein time-discretization schemes. Since the SABR model imposes an absorbing boundary condition at the origin, the numerical boundary condition in the time discretization schemes have to be more carefully implemented. It has been well known that naive truncation at the origin may lead to large deviation from the original SABR model. Typically, it is necessary to use a small time discretization step to decrease the discretization errors due to its low order of convergence, thus causing heavy computation load. 

Several simulation algorithms~\citep{chen2012low,cai2017sabr,leitao2017one,leitao2017multi,
cui2018ctmc-sabr,grzelak2019stochastic,kyriakou2023unified} have been proposed for the simulation of the SABR model. Simulation of the SABR model over one time step requires sequential sampling procedures of the terminal volatility, average variance, and terminal forward price at the next simulation time point. Since the last two sampling procedures are challenging, existing methods distinguish among themselves by adopting different choices of the algorithms for these two steps. These earlier simulation algorithms have achieved some successes to overcome the difficulty of imposing the absorbing boundary condition in the time-discretization methods. This paper aims to provide more efficient, accurate and reliable simulation scheme that competes favorably with the existing simulation schemes in the literature.

Regarding the sampling of the average variance conditional on the terminal volatility, \cite{chen2012low} sampled this quantity through the mean-and-variance matched lognormal (LN) variable. While the sampling procedures are fast, the time step cannot be large since the derived mean and variance are valid in small-time limit only. For more accurate sampling, later works employed various versions of the transform of the average variance. \cite{leitao2017one,leitao2017multi} derived the approximate Fourier transform of the unconditional average variance recursively, while \citet{cai2017sabr} used the analytic formula of the Laplace transform of the reciprocal of the conditional average variance~\citep{matsuyor2005exp1}. However, these approaches are cumbersome and time-consuming since the tedious numerical inverse transform procedures are required to construct the cumulative distribution function (CDF) and implement the iterative root-finding for the inverse transform sampling. As an alternative approach, \cite{cui2021efficient} approximated the volatility dynamics with the continuous-time Markov chain (CTMC) over volatility grids. Unfortunately, their approach does not immune from heavy computation and cumbersome implementation. In another approach, \citet{kyriakou2023unified} used the moment-matched Pearson family of distributions for efficient sampling. For such purpose, they rely on numerical evaluation of the four order moments from the Laplace transform of conditional average variance. 

Sampling the terminal forward price conditional on the terminal variance and average variance is much more challenging. Unlike the other stochastic volatility models that are based on the geometric Brownian motion (BM), the constant-elasticity-of-variance (CEV) feature of the SABR model does not admit an analytically tractable form of the conditional terminal forward price distribution. As a result, the price distribution has to be approximated in some analytic form first before the corresponding sampling algorithm is considered. For the approximation of the conditional terminal forward price distribution, most existing algorithms~\citep{chen2012low,cai2017sabr,leitao2017one,leitao2017multi,cui2018ctmc-sabr,grzelak2019stochastic,kyriakou2023unified} invariably adopted \citet{islah2009sabr-lmm}'s noncentral chi-squared (NCX2) distribution approximation. However, the failure of the martingale property of the conditional forward price in this approximation has not been critically examined among these earlier schemes. \cite{leitao2017one,leitao2017multi} applied an ad-hoc correction on the simulation prices to enforce the martingale property. Our paper proposes the CEV approximation of the conditional terminal forward price distribution, which is considered to be better since it preserves the martingale property and provides more efficient simulation procedures. 

Indeed, the lack of fast sampling algorithm for the approximated conditional forward price distribution is the last challenge in constructing an efficient SABR simulation algorithm. The inverse transform sampling with numerical root-finding procedure is accurate but tediously slow. The mean-and-variance-match quadratic Gaussian sampling has been adopted in \citet{chen2012low}, \citet{leitao2017one,leitao2017multi}, and \citet{cui2021efficient} to speed up sampling. However, it is limited to small time step due to the nature of approximation. Also, it cannot be used efficiently when the price is close to zero due to the absorbing boundary condition. 

This paper proposes an accurate, efficient, and reliable SABR simulation algorithm, filling the gaps in all steps of the simulation procedure. The contributions of this paper are three-fold. Firstly, we derive the first four moments of the conditional average variance analytically. We then use the shifted lognormal (SLN) distribution obtained from matching the moments for fast sampling of the average variance. Since we match the moments that are in higher order and exact for any arbitrary time step, good accuracy of our SLN distribution approximation can be maintained at larger time step.
Secondly, we present an alternative approximation of the conditional forward  price distribution based on the CEV process. Our approximation shows better accuracy when compared with the widely adopted Islah's approximation \citep{islah2009sabr-lmm} since it preserves the martingale property of the conditional forward price by its construction. Lastly, we employ the algorithms of \citet{makarov2010exact} and \citet{kang2014simulation} that are capable of sampling the CEV process exactly. Their algorithms are fast since the CEV distribution can be expressed by a compound random variable composed of three elementary (namely, gamma-Poisson-gamma) random variables.

This paper is organized as follows. Section~\ref{s:model} presents the formulation of the SABR model and reviews some of its analytic properties. Section~\ref{s:avg_var} describes our SLN approximation for sampling the average variance and presents its performance in numerical accuracy. In Section~\ref{s:cond_price}, we present the CEV distribution approximation of the conditional forward price, and an exact sampling algorithm for the CEV distribution via  the shifted Poisson-mixture Gamma distribution. In particular, we argue that our CEV distribution approximation of the conditional forward price is better than the widely adopted Islah's approximation as it preserves the martingale condition. Section~\ref{s:num} presents our comprehensive numerical experiments that serve to assess accuracy, efficiency and reliability of different simulation schemes and analytic approximation methods, displaying numerical comparisons under the more challenging choices of parameters, such as large time-to-maturity. Comparison of accuracy performance of our CEV approximation of the forward price with that of the Islah's approximation is carefully examined. Conclusive remarks are summarized in Section~\ref{s:conc}.

\section{Formulation of the SABR model and its two-step simulation}\label{s:model}\noindent
In this section, we introduce the SABR model and review some of its analytic properties under various degenerate cases. The two steps of the SABR simulation are outlined. Specifically, we discuss some special cases where the SABR simulation becomes straightforward.

\subsection{SABR model} \noindent
The governing stochastic differential equations (SDEs) for the SABR volatility model~\citep{hagan2002sabr} are given by
\begin{equation}\label{eq:sde}
	\frac{\mathrm{d}F_t}{F_t^\beta} = \sigma_t  \, \mathrm{d}W_t \qtext{and} \frac{\mathrm{d} \sigma_t}{\sigma_t} = \vov \, \mathrm{d}Z_t,\quad 0 \leq \beta \leq 1, 
\end{equation}
where $F_t$ and $\sigma_t$ are the stochastic processes for the forward price and volatility, respectively, $\vov$ is the vol-of-vol, $\beta$ is the elasticity of variance parameter, and $W_t$ and $Z_t$ are the standard Brownian motions (BMs) that are correlated with correlation coefficient $\rho$.
We also define two parameters for notational simplicity:
$$ \rhoc := \sqrt{1-\rho^2} \qtext{and} \betac := 1 - \beta.
$$

\subsection{Average variance and two-step simulation procedures} \noindent
Regarding the simulation over one time step with size $h$, we sample the volatility and forward price at time $t+h$ ($\sigma_{t+\dt}$ and $F_{t+\dt}$) for discrete time step $\dt$ at given time $t$. Accordingly, the filtration up to time $t$ (say, $\sigma_t$ and $F_t$) is implicitly assumed. Since the volatility process $\sigma_t$ in Eq.~\eqref{eq:sde} follows a geometric BM, we have $\sigma_t = \sigma_0\exp\left(\vov Z_t - \vov^2t/2 \right).$ Simulating $\sigma_t$ from time $t$ to $t+\dt$ can be easily performed via
\begin{equation} \label{eq:sig_t}
\sigma_{t+\dt} = \sigma_t\exp\left(\vov Z_h - \frac{\vov^2 h}{2} \right) =  \sigma_t\exp\left(\vovn \hat{Z} \right)\quad
\qtext{where} \hat{Z} \sim \mathcal{N}\left(-\frac{\vovn}{2}, 1\right).
\end{equation}
Here, $\vovn := \vov\sqrt{h}$
is the standard deviation of $\ln \sigma_t$ over the time step $h$. The quantity $\vovn$ is used as a measure of \textit{small-time} limit.

Like the simulation schemes in other stochastic volatility models, the average variance conditional on initial and terminal volatility values plays an important role in the SABR model simulation. Given the initial and terminal volatility (namely, $\sigma_t$ and $\sigma_{t+\dt}$), we define the dimensionless conditional average variance  between $t$ and $t+h$ as
\begin{equation} \label{eq:normalized}
	\IVstep (\sigma_{t+\dt}) := \frac1{\sigma_t^2 h} \int_{t}^{t+\dt} \sigma_s^2\, \mathrm{d}s \;\Big|_{\sigma_{t+\dt}}.
\end{equation}
Here, we normalize the integrated variance by $\sigma_t^2 h$ in order that $\IVstep$ becomes a dimensionless quantity of order one. Specifically, $\IVstep$ converges to one in the limit of $\vov\downarrow 0$. From the simulation step in Eq.~\eqref{eq:sig_t}, $\IVstep$ is equivalently conditioned by $\hat{Z}$ as
$$ \IVstep(\hat{Z})
\;\sim\; \int_{0}^{1} e^{2\vovn Z_s}\, \mathrm{d}s \;\Big|_{Z_1 = \hat{Z}}\;.
$$
For notational simplicity, we simply write $\IVstep$ by omitting the dependence on $\sigma_{t+\dt}$ or $\hat{Z}$.

Almost all SABR simulation algorithms strive to perform the following two steps of simulation to achieve efficiency and high accuracy, given the filtration up to $t$:
\begin{description}
	\item[Step 1]: simulation of the average variance $\IVstep$ conditional on $\sigma_{t+\dt}$ (or $\hat{Z}$), 
	\item[Step 2]: simulation of the forward price $F_{t+\dt}$ conditional on $\sigma_{t+\dt}$ and $\IVstep$. 
\end{description}
This paper aims to innovate these two steps of simulation over existing algorithms, details are presented in Sections~\ref{s:avg_var} and \ref{s:cond_price}, respectively.

Regarding Step 2, we define the conditional expectation of $F_{t+\dt}$:
\begin{equation}\label{eq:cond_Ft_h}
	\condFstep(\sigma_{t+\dt}, \IVstep) := E\left(F_{t+\dt} \,|\, \sigma_{t+\dt}, \IVstep\right),
\end{equation}
which is essential for the discussion of the simulation algorithms of Step 2. For notational simplicity, we omit the dependence on $\sigma_{t+\dt}$ and $\IVstep$, and simply write $\condFstep$.

\subsection{Special cases} \label{ss:special} \noindent
Before introducing our simulation algorithm, we review several special cases of the SABR model, which helps building intuition for our new algorithm.

\textbf{Normal ($\beta=0$) SABR model:}
When $\beta=0$, conditional on $\sigma_T$ and $\IV$, integrating $F_t$ from $t=0$ to $T$ reveals that the terminal forward price $F_T$ follows the normal distribution with mean $\condF$ and  variance $\rhoc^2\sigma_0^2 T\IV$:
$$ F_T \,|\, \sigma_T, \IV \;\sim\; \mathcal{N}\left(\condF, \rhoc^2 \sigma_0^2 T\IV\right), $$
where the conditional expectation $\condF$ is given by
\begin{equation} \label{eq:condF_norm}
	\condF = F_0 + \frac{\rho}{\vov}\big( \sigma_T - \sigma_0  \big).
\end{equation}
Therefore, Step 2 becomes trivial as long as $\IVstep$ can be sampled accurately in Step 1. \citet{choi2019nsvh} managed to find an exact closed-form expression for sampling $F_T$ without separating Steps 1 and 2. Moreover, note that $F_T$, as well as $\condF$ can be negative under normal distribution as the origin is not a boundary. Accordingly, the normal SABR model is different from the $\beta\downarrow 0$ limit of the SABR model. In the later case, the origin is an absorbing boundary. In view of these subtleties, we discard further discussion of the SABR model under $\beta=0$.

\textbf{Lognormal ($\beta=1$) SABR model:} When $\beta=1$, conditional on $\sigma_T$ and $\IV$, integrating $\ln F_t$ from $t=0$ to $T$ reveals that the terminal forward price $F_T$ follows the LN distribution with mean $\condF$ and log variance $\rhoc^2\sigma_0^2 T\IV$:
\begin{equation} \label{eq:beta1}
	F_T \,|\, \sigma_T, \IV \;\sim\; \condF\, \exp\left(\,\rhoc \sigma_0 \sqrt{T\IV}\, X - \frac{\rhoc^2 \sigma_0^2 T \IV}2 \right),
\end{equation} 
where $X$ is a standard normal variate independent of $\sigma_T$ and $\IV$. The conditional  expectation $\condF$ is explicitly given by
\begin{equation} \label{eq:Ft_cond1}
	\condF = F_0 \exp\left(\frac{\rho}{\vov}(\sigma_T - \sigma_0)-\frac{\rho^2 \sigma_0^2 T \IV}{2} \right).
\end{equation}
Therefore, Step 2 becomes trivial as well when $\beta=1$.

This special case of $\beta=1$ provides the important intuition for our CEV approximation of the forward price under $0<\beta<1$ (to be discussed in Section~\ref{s:cond_price}).
\begin{remark} \label{r:total_exp}
From the law of total expectation, it follows that
\begin{equation*}
	E(\condF) = E\left(E\left(F_T|\sigma_T,\IV\right)\right) = E(F_T) = F_0,
\end{equation*} and the joint distribution of $\sigma_T$ and $\IV$ satisfies
$$
E\left[\exp\left(\frac{\rho}{\vov}(\sigma_T - \sigma_0)-\frac{\rho^2 \sigma_0^2 T \IV}{2}\right)\right] = 1,
$$
for any initial volatility $\sigma_0>0$.
\end{remark}

\textbf{Zero vol-of-vol ($\vov= 0$) SABR model:} For $0<\beta<1$, the SABR model is now reduced to the CEV model:
\begin{subequations}
	\begin{equation} \label{eq:cev-sde}
		\frac{\mathrm{d}F_t}{F_t^\beta} = \sigma_0 \, \mathrm{d}W_t,
	\end{equation}
	with the absorbing boundary condition at the origin. Here, $\sigma_0$ is constant. Indeed, the SABR model can be visualized as the stochastic volatility version of the CEV model.
	\begin{definition}
		The CEV distribution is defined as the distribution of $F_T$ resulting from the CEV dynamics as depicted in Eq.~\eqref{eq:cev-sde}, which is denoted by
		\begin{equation}\label{eq:def_CEV}
			F_T \sim \cev(F_0, \sigma_0^2T).
		\end{equation}
		It is parameterized by the elasticity-of-variance parameter $\beta$, mean $F_0$, and variance $\sigma_0^2 T$. The cumulative distribution function (CDF) and probability density function (PDF) of the CEV distribution are presented respectively in Eqs.~\eqref{eq:cev_CDF} and ~\eqref{eq:cev_PDF} later.
	\end{definition}
\end{subequations}
If $F_T\sim \cev(F_0, \sigma_0^2 T)$, then $E(F_T) = F_0$ regardless of the values of $\beta$ and $\sigma_0^2 T$. In the zero vol-of-vol ($\vov\downarrow 0$) limit, $F_T$ from the SABR model follows $\cev(F_0, \sigma_0^2T)$ unconditionally. Step 2 is reduced to sampling $F_T$ from the CEV distribution.

\textbf{Uncorrelated ($\rho=0$) SABR model:} The terminal forward price $F_T$ also exhibits a CEV distribution under non-zero vol-of-vol ($\vov>0$) and zero correlation ($\rho=0$) case, albeit conditionally. Conditional on $\IV$, the terminal price $F_T$ follows the CEV distribution:
$$ F_T\,|\, \IV \sim \cev(F_0, \sigma_0^2 T \IV).$$ 
Basically, $T \IV$ plays the role of stochastic time clock of the CEV model. Like the $\vov=0$ case, Step 2 is reduced to sampling $F_T$ from the CEV distribution under $\rho=0$. Unlike the $\beta=1$ case, $F_T$ is only conditioned by $\IV$ since conditional expectation is always $F_0$ without depending on $\sigma_T$. Thus, $\IV$ should be understood as the unconditional average variance in this context. Therefore, the European option price can be expressed as one dimensional integral of the CEV option price over the (unconditional) distribution of $\IV$. As a remark, \citet{choi2021note} approximated the integral using the Gaussian quadrature together with the LN distribution fitted to the (unconditional) mean and variance of $\IV$.

Besides the above four special cases, one has to face with the general case with $0<\beta<1$ and $\rho\neq 0$, under which the conditional distribution of $F_T$ is not analytically tractable. The success of a simulation algorithm for the SABR model amounts to an ingenious choice chosen for the approximation of the conditional distribution of $F_T$.

\section{First simulation step: Sampling average variance $\IVstep$} \label{s:avg_var} \noindent
This section presents our simulation method for sampling $\IVstep$. We also compare with other most popular simulation methods, highlighting the computational advantages of our simulation scheme.
\subsection{Sampling $\IVstep$ from the moment-matched shifted lognormal (SLN) approximation} \label{ss:sigma_vt_new}\noindent
We sample $\IVstep$ from the SLN distribution matched to the first three moments of $\IVstep$.

\begin{definition} The SLN random variable, $Y\sim \mathcal{SLN}(\mu, \sigma^2, \lambda)$, is given by
	$$ Y\sim\mu \left[(1-\lambda) + \lambda \exp\left(\sigma X - \frac{\sigma^2}{2}\right)\right] \qtext{with} \mu>0, \; \sigma>0, \; 0 < \lambda\le 1,$$
	where $X$ is a standard normal variate. Here, $\mu$ is the mean of $Y$, $\lambda$ and $\sigma^2$ are the shifted parameter and variance of the LN component, respectively. When $\lambda=1$, $Y$ is reduced to the LN variable with mean $\mu$ and log variance $\sigma^2$.
	
	The coefficient of variation, skewness, and ex-kurtosis of $Y$ are given by
	$$v(Y) = \lambda \sqrt{w},\quad \text{s}(Y)=\sqrt{w}(w+3), \qtext{and} \text{k}(Y) = w(w^3+6w^2+15w+16),$$
	respectively, where $w=e^{\sigma^2}-1\;(w>0)$. 
\end{definition}
\begin{remark}
	While the skewness and ex-kurtosis of distribution $Y$ follow the conventional definitions, the coefficient of variation is defined as the ratio of the standard deviation to mean:
	$$ v(Y) = \frac{\sqrt{\mathrm{Var}(Y)}}{E(Y)}.$$
	The coefficient of variation is an important dimensionless quantity that characterizes non-negative distributions such as the LN and SLN distributions.
\end{remark}

\begin{proposition}[SLN distribution matched to the first three moments] \label{p:sln}
	Given mean $\mu>0$, coefficient of variance $v>0$, and skewness $s>0$, the parameters of the SLN distribution matching $\mu$, $v$, and $s$ admit the following analytic formulas:
	\begin{gather*}
		\sigma = \sqrt{\ln (1+w)} \qtext{for} w = 4\sinh^2\left(\frac16\acosh\left(1+\frac{s^2}{2}\right)\right), \quad
		\text{and}\quad\lambda=\frac{v}{2\sinh\left(\frac16\acosh\left(1+\frac{s^2}{2}\right)\right)}.		
	\end{gather*}
\end{proposition}
\begin{proof}
We first fit $\sigma$ to the given skewness $s$: 
$$s^2=w(w+3)^2 \qtext{for} w = e^{\sigma^2}-1.$$
Substituting $x=w+2$ transforms the equation to a special case of the cubic equation: 
$$ x^3 - 3x - (2+s^2) = 0, $$
where the unique positive root $x$ is analytically given by\footnote{See \url{https://en.wikipedia.org/wiki/Cubic_equation\#Hyperbolic_solution_for_one_real_root}.}
$$ x = 2\cosh\left(\frac13 \acosh\left(1+\frac{s^2}{2}\right)\right).
$$
The expression for $w$ follows from $\cosh \theta - 1= 2\sinh^2\frac{\theta}{2}$.
Next, we fit $\lambda$ to the given coefficient of variance $v$:
$$\lambda = \frac{v}{\sqrt{w}}=\frac{v}{2\sinh\left(\frac16\acosh\left(1+\frac{s^2}{2}\right)\right)}.
$$
\end{proof}
\begin{remark}
	For an LN random variable (i.e., $\lambda=1$), $\sigma$ is determined by matching the coefficient of variation $v$:
	$$ \sigma=\sqrt{\ln\left(1\,+\,v^2\right)}. $$
	\citet{chen2012low} has sampled $\IVstep$ through the LN random variable fitted to the mean and variance of $\IVstep$ in the small-time limit.
\end{remark}

Next, we derive the first four conditional moments of $\IVstep$, whose analytic formulas are presented in Proposition~\ref{p:cond_mom}.
\begin{proposition}[Conditional moments of average variance] \label{p:cond_mom}
Given $\sigma_t$ and $\sigma_{t+\dt}$, the first four raw moments of $\IVstep$ (i.e., $\mu = E(\IVstep)$ and $\mu'_k = E\left((\IVstep)^k\right)$) are given by
\begin{align*}
	\mu &= \left(\frac{\sigma_{t+\dt}}{\sigma_t}\right) m_1,\\
	\mu'_2 &= \left(\frac{\sigma_{t+\dt}}{\sigma_t}\right)^2 \frac{1}{\vovn^2} \left(m_2 - c m_1 \right),
	\\
	\mu'_3 &= \left(\frac{\sigma_{t+\dt}}{\sigma_t}\right)^3 \frac{1}{8\vovn^4}\left[3 m_3 - 8c m_2 + \left(4c^2 + 1\right)m_1 \right],	\\
	\mu'_4 &= \left(\frac{\sigma_{t+\dt}}{\sigma_t}\right)^4 \frac{1}{24\,\vovn^6} \left[2 m_4 - 9c  m_3 + \left(12c^2 + 2\right) m_2 - c\left(4c^2 + 3\right)m_1\right],
\end{align*}
where the conditioning variable, $\sigma_{t+\dt}$ and $\hat{Z}$, are exchangeable by $\hat{Z} = \frac{1}{\vovn}\ln (\sigma_{t+\dt}/\sigma_t)$, and the coefficients, $c$ and $m_k\, (k=1,2,3,4)$, are functions of $\sigma_t$, $\sigma_{t+\dt}$ and $\hat{Z}$:
$$ 
c(\hat{Z}) := \cosh\vovn \hat{Z} = \frac12 \left(\frac{\sigma_{t+\dt}}{\sigma_t} + \frac{\sigma_t}{\sigma_{t+\dt}}\right) \qtext{and}
m_k(\hat{Z}) = \frac{N(\hat{Z}+k\vovn)-N(\hat{Z}-k\vovn)}{2k\vovn\; n\left(\sqrt{\hat{Z}^2+(k\vovn)^2}\right)},
$$
and $n(z)$ and $N(z)$ are the PDF and CDF of the standard normal distribution, respectively.
\end{proposition}
\begin{proof}
	The analytic derivation is based on the conditional moments of the exponential functional of Brownian motion~\citep[(5.4)]{matsuyor2005exp1}. See \ref{apdx:cond_mom} for the derivation.
\end{proof}
\begin{remark}
	\citet[Eqs.~(7) and (8)]{kennedy2012prob} have derived $\mu$ and $\mu'_2$ independently without resorting to \citet{matsuyor2005exp1}.
\end{remark}
\begin{remark} \label{r:mom}
	From the first four raw moments, the coefficient of variation $v$, skewness $s$, and ex-kurtosis $\kappa$ of $\IVstep$ are obtained as
	$$v(\IVstep) = \frac{\sqrt{\mu_2}}{\mu}, \quad s(\IVstep) = \frac{\mu'_3 - 3\mu \mu'_2 + 2\mu^3}{\mu_2\sqrt{\mu_2}} ,\qtext{and} \kappa(\IVstep) = \frac{\mu'_4 - 4\mu\mu'_3 + 6\mu^2\mu'_2 - 3\mu^4}{\mu_2^2} - 3,
	$$
	where $\mu_2 = \mu'_2 - \mu^2$ is the variance.
\end{remark}

From Propositions~\ref{p:sln} and \ref{p:cond_mom}, it is possible to sample $\IVstep$ from the SLN variable whose parameters are fitted to the first three moments of $\IVstep$ exactly. However, we take a much simpler computational approach that uses a fixed $\lambda$ obtained from the small-time limit. It is shown to be effective as well (see Algorithm~\ref{alg:alg1_IV}). 
%
In \ref{apdx:cond_mom}, the coefficient of variation and skewness of $\IVstep$ are expanded around $\vovn=0$ as 
$$ v = \frac{\sqrt{\mu_2}}{\mu} = \frac{\vovn}{\sqrt{3}} + O(\vovn^3)\qtext{and} s = \frac{\mu_3}{\mu_2\sqrt{\mu_2}} = \frac{6\sqrt{3}}{5} \vovn + O(\vovn^3),\quad\text{where }\vovn=\vov\sqrt{h}
$$
Using Taylor's expansion, we obtain
$$ 2\sinh\left(\frac16\acosh\left(1+\frac{s^2}{2}\right)\right)=\frac{s}{3}-\frac{s^3}{81}+O(s^5),$$
so that the shift parameter $\lambda$ converges to
$$ \lambda=\frac{v}{2\sinh\left(\frac16\acosh\left(1+\frac{s^2}{2}\right)\right)}\to \frac56 \qtext{as} \vovn\downarrow 0.$$
Once $\lambda$ is determined, we fit $\sigma$ to match the coefficient of variation $v$:
$$ \sigma=\sqrt{\ln\left(1+\frac{v^2}{\lambda^2}\right)}=\sqrt{\ln\left(1+\frac{36}{25}v^2\right)}.$$
As we demonstrate in the numerical results shown in Figure~\ref{f:moments}, this approximation is highly reliable not only for small $\vovn$ but even for $O(\vovn)=1$.

We summarize the above results into Algorithm~\ref{alg:alg1_IV} for approximating conditional $\IVstep$ as follows:


\vspace{1em}
\begin{algorithm}[H] 
	\caption{Approximation of the conditional $\IVstep$} \label{alg:alg1_IV}
	In the small-time limit ($\vovn\downarrow 0$), $\IVstep$ is approximately sampled via an SLN variable with $\lambda=5/6$:
	\begin{equation}
		\label{eq:shifted_sigma}
		\IVstep \sim \frac{\mu}{6} \left[1 +  5 \exp\left(\sigma X - \frac{\sigma^2}{2}\right)\right]
		\qtext{with} \sigma=\sqrt{\ln\left(1+\frac{36}{25}v^2\right)},
	\end{equation}
	where $\mu$ and $v$ are the mean and coefficient of variance of $\IVstep$ given by Proposition~\ref{p:cond_mom} and Remark~\ref{r:mom}.
\end{algorithm}
\vspace{1em}

\begin{figure}[H]
	\includegraphics[width=0.5\linewidth]{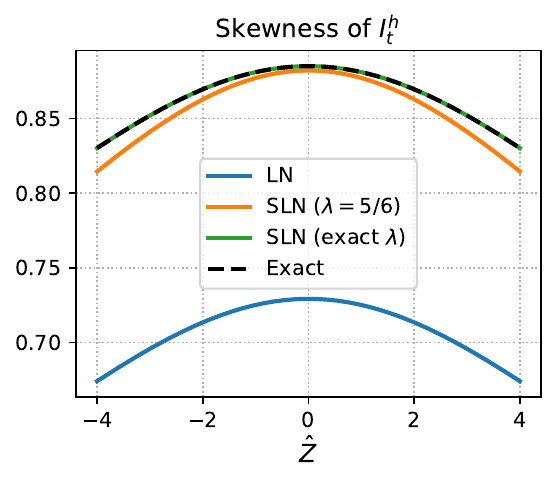}
	\includegraphics[width=0.5\linewidth]{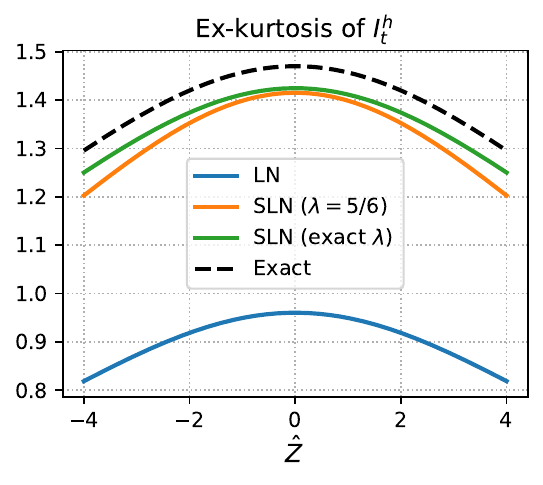}
	\caption{Comparison of accuracy of calculating the skewness (left) and ex-kurtosis (right) of $\IVstep$ as functions of $\hat{Z}$ for $\vovn = 0.4$. We compare our SLN approximation with $\lambda=5/6$ (orange), the SLN approximation fitted up to skewness (green), and the LN approximation fitted up to variance (blue) to the true values (dashed).} 
	\label{f:moments}
\end{figure}

\subsection{Comparison to other simulation methods in literature}\label{ss:compare_IV}\noindent
Our sampling method for the approximation of conditional $\IVstep$ is much faster and easier to implement compared to other methods in the literature, such as \citet{chen2012low}, \citet{cai2017sabr}, \citet{leitao2017one,leitao2017multi} and \citet{cui2021efficient}.

Our method is considered as an extension of \citet{chen2012low}. Their LN approximation is matched up to mean and variance of $\IVstep$, while our approximation is matched up to the first three moments. Furthermore, they derive the mean and variance only in the small-time limit, so the time step has to be chosen to be small accordingly. We enhance the method by deriving the moments analytically and matching to higher moments with the SLN approximation. Figure~\ref{f:moments} displays accuracy of skewness and ex-kurtosis of $\IVstep$ obtained from our SLN approximation. The left figure shows that the SLN approximation (green curve) gives the correct skewness value across the dependency of $\hat{Z}$ since the parameters are fitted up to skewness. Moreover, the skewness under the simple SLN scheme with $\lambda=5/6$ (orange curve) also agrees well. Even though the two SLN schemes do not use ex-kurtosis in the parameter calibration, the SLN approximation of the ex-kurtosis of $\IVstep$ is close to the true value. This justifies that our chosen SLN approximation (either exact $\lambda$ or $\lambda=\frac56$) is a good choice of approximation of $\IVstep$ (see right figure). As evidenced from the plots in Figure~\ref{f:moments}, the SLN approximation (either exact $\lambda$ or $\lambda=\frac56$) outperforms the LN approximation.

\citet{cai2017sabr} computed the analytic form of  the Laplace transform of the CDF of the reciprocal of the average variance 1/$\IVstep$, then used the Euler inversion formula to obtain the CDF of 1/$\IVstep$. After obtaining the CDF, they used the inverse transform method, which requires some tedious root-finding procedures such as Newton's method or bisection method to sample $\IVstep$. 

\citet{leitao2017one,leitao2017multi} used recursion schemes to find the characteristic function of the logarithm of $\IVstep$ and then employed the inverse Fourier technique COS method \citep{fang2008novel} to recover its PDF. Next they computed the Pearson correlation coefficient of $\ln\IVstep$ and $\ln\sigma_{t+\dt}$, then the bivariate copula distribution based on the CDF of $\ln\sigma_T$ and approximate CDF of $\ln\IVstep$. Finally, they adopted the direct inverse transform method based on linear interpolation to sample $\IVstep$. 

\citet{cui2021efficient} used the continuous-time Markov chain (CTMC) in the space of the volatility process with finite number of states. In their procedure, they first fixed the volatility grids, computed the transition probability density at these grids based on the exponential of the generator matrix and used the inverse transform method to sample the volatility. After obtaining the volatility values, they computed the conditional characteristic function of the integrated variance, again based on the exponential of the generator matrix. They then adopted the Fourier sampler to get the CDF of the integrated variance, by first computing the PDF of the integrated variance via the orthogonal projection. Finally, they simulated the integrated variance using the closed-form inverse function. 

Besides the issue of accuracy of the simulation of $\IVstep$, these simulation methods invariably include tedious and time-consuming procedures such as the adoption of the inverse Laplace/Fourier transform calculation, root-finding steps in the inverse transform methods or computation of the exponential of matrices. However, our SLN simulation methods for $\IVstep$ exhibits easy and fast implementation. As shown in the numerical experiments reported in Section~\ref{s:num}, our simulation methods for approximating $\IVstep$ demonstrate significant savings in computational time and high level of accuracy.

\section{Second simulation step: Sampling forward price $F_{t+\dt}$} \label{s:cond_price}\noindent
This section presents our CEV approximation for the forward price $F_{t+\dt}$ and the sampling of the CEV distribution via the exact simulation scheme of the shifted Poisson mixture gamma variate.
\subsection{CEV approximation of the conditional forward price}\label{ss:our}\noindent
In the uncorrelated SABR model ($\rho=0$), we observe that the conditional forward price $F_T$ follows the CEV distribution. To construct our SABR simulation scheme, we propose the CEV approximation of the conditional $F_T$ under the general case of $\rho\neq 0$ and $0<\beta<1$. We employ the operator splitting technique to derive the CEV approximation. Recall that a similar operator splitting approach has been used in the design of the simulation schemes for the rough Heston model in \cite{BayerBreneis2024}.

We rewrite the SDEs for the SABR volatility model [See Eq.~\eqref{eq:sde}] into the following form:
\begin{equation}\label{eq:sde_new}
	\frac{\text{d}F_t}{F_t^{\beta}}=\sigma_t(\rho\,\text{d}Z_t+\sqrt{1-\rho^2}\,\text{d}Z_t^{\perp})\quad\text{and}\quad\frac{\text{d}\sigma_t}{\sigma_t}=\vov\,\text{d}Z_t,
\end{equation}
where $Z_t$ and $Z_t^{\perp}$ are uncorrelated BMs. Similar to the operator splitting approximation method in \cite{LileikaMackevicius2021}, we split the above SDEs into two parts:
\begin{subequations}
	\begin{equation}
		\frac{\text{d}F_t^{(1)}}{[F_t^{(1)}]^{\beta}}=\rho\sigma_t\,\text{d}Z_t\quad\text{and}\quad\frac{\text{d}\sigma_t}{\sigma_t}=\vov\,\text{d}Z_t
		\label{eq:splitting_first_step}
	\end{equation}
	\text{and}
	\begin{equation}
		\frac{\text{d}F_t^{(2)}}{[F_t^{(2)}]^{\beta}}=\sqrt{1-\rho^2}\sigma_t\,\text{d}Z_t^{\perp}\quad\text{and}\quad\frac{\text{d}\sigma_t}{\sigma_t}=\vov\,\text{d}Z_t.
		\label{eq:splitting_second_step}
	\end{equation}
\end{subequations}
The operator splitting approximation method dictates that an approximation to the solution of $F_t$ in Eq.~\eqref{eq:sde_new} with initial condition $F_0$ is given by $F_t\sim F_t^{(2)}\left(F_t^{(1)}(F_0)\right)$, where $F_t^{(1)}(F_0)$ is the solution of $F_t^{(1)}$ in Eq.~\eqref{eq:splitting_first_step} with the initial condition $F_0^{(1)}=F_0$ and $F_t^{(2)}\left(F_t^{(1)}(F_0)\right)$ is the solution of $F_t^{(2)}$ in Eq.~\eqref{eq:splitting_second_step} with the initial condition $F_0^{(2)}=F_t^{(1)}(F_0)$.

Next, we solve the simplified set of SDEs for $F_t^{(1)}$ and $F_t^{(2)}$. The SDEs in Eq.~\eqref{eq:splitting_second_step} resembles the SABR model with $\rho=0$ so that $F_{t}^{(2)}|\,\sigma_{t},I_{0}^t$ follows the CEV distribution, where
\begin{equation}\label{eq:F_t_2}
	F_{t}^{(2)}|\,\sigma_{t},I_{0}^t \;\sim\; \cev(F_0^{(2)},\rhoc^2\sigma_0^2 \dt I_0^t),\quad\rhoc^2=1-\rho^2.
\end{equation}
However, there is no simple analytic solution to Eq.~\eqref{eq:splitting_first_step} for $\rho\neq 0$ and $0<\beta<1$. We propose to employ the frozen coefficient approximation to approximate $[F_t^{(1)}]^{\betac}$ by $F_0^{\betac}$ such that Eq.~\eqref{eq:splitting_first_step} is approximated by
\begin{equation}\label{eq:cev_gbm}
	\frac{\text{d}F_t^{(1)}}{F_t^{(1)}}=\frac{\rho\sigma_t}{[F_t^{(1)}]^{\betac}}\,\text{d}Z_t\approx\frac{\rho\sigma_t}{F_0^{\betac}}\,\text{d}Z_t.
\end{equation}
Let $F_t^{(a)}$ denote the solution of the above approximating SDE, which is considered as an approximation of $F_t^{(1)}$. By rewriting the associated SDE for $F_t^{(a)}$ into the form:
\begin{equation*}
	\text{d}\ln F_t^{(a)}=\frac{\rho}{\vov F_0^{\betac}}\,\text{d}\sigma_t-\frac{\rho^2\sigma_t^2}{2F_0^{2\betac}}\,\text{d}t,
\end{equation*}
the solution of $F_t^{(a)}$ can be easily found by direct integration, which gives
\begin{equation}\label{eq:first_appr_split} 			   F_t^{(a)}(F_0)=F_0\exp\left(\frac{\rho}{\vov F_0^{\betac}}(\sigma_t-\sigma_0)-\frac{\rho^2\sigma_0^2 t I_0^t}{2F_0^{2\betac}}\right).
\end{equation}
By combining the above operator splitting procedure coupled with frozen coefficient approximation, the approximation solution of $F_t$ is obtained as $F_t^{(2)}\left(F_t^{(a)}(F_0)\right)$. Conditional on $\sigma_t$ and $I_0^t$, we obtain the CEV approximation of the conditional terminal forward price $F_t$ in the SABR model as follows:
\begin{subequations}
	\begin{equation}\label{eq:cev_approximation_cev}
		F_{t}\,|\, \sigma_{t}, I_0^t \;\sim\; \cev(\bar{F}^t_0, \rhoc^2\sigma_0^2 t I_0^t),
	\end{equation}
	where
	\begin{equation}\label{eq:conditional_mean_F_T_bar}
		\bar{F}^t_0 = F_0\exp\left(\frac{\rho}{\vov F_0^{\betac}}(\sigma_t-\sigma_0)-\frac{\rho^2\sigma_0^2 t I_0^t}{2F_0^{2\betac}}\right).
	\end{equation}
\end{subequations}

The above discussion considers the approximation of solution of $F_t$ with initial time zero and end time $t$. In our context, we solve for $F_{t+h}$ with initial time $t$ and end time $t+h$. Accordingly, we change the initial time zero to time $t$ (comparing Eqs.~(\ref{eq:cev_approximation_cev}, \ref{eq:conditional_mean_F_T_bar}) with Eqs.~(\ref{eq:beta01}, \ref{eq:Ft_cond_new})). The CEV approximation of the conditional forward price $F_{t+\dt}\,|\, \sigma_{t+\dt}, \IVstep$ is summarized in Algorithm~\ref{alg:alg_cev}:

\vspace{1em}
\begin{algorithm}[H] 
	\setlength{\abovedisplayskip}{10pt}
	\setstretch{1}
	\caption{CEV approximation of the conditional terminal forward price} \label{alg:alg_cev}
	For $0<\beta<1$ and $-1 \le \rho\le 1$, conditional on $\sigma_{t+\dt}$ and  $\IVstep$ (filtration up to time $t$), we approximate $F_{t+\dt}$ as a CEV distribution as follows:
	\begin{subequations}
		\begin{equation} \label{eq:beta01}
			F_{t+\dt}\,|\, \sigma_{t+\dt}, \IVstep \;\sim\; \cev(\condFstep, \rhoc^2\sigma_t^2 \dt \IVstep),
		\end{equation}
		where the conditional expectation $\condFstep$ is approximated by
		\begin{equation} \label{eq:Ft_cond_new}
			\condFstep \approx F_t \exp\left(\frac{\rho(\sigma_{t+\dt} - \sigma_t)}{\vov F_t^{\betac}}-\frac{\rho^2 \sigma_t^2 \dt \IVstep}{2 F_t^{2\betac}} \right).
		\end{equation}
	\end{subequations}
\end{algorithm}
\vspace{1em}

Note that our choice of $\condFstep$ ensures the martingale property of conditional terminal forward price $F_{t+h}$. 
Based on Remark~\ref{r:total_exp}, we observe
\begin{equation}\label{eq:F_T_bar}
	E(F_{t+\dt}) = E\left(E(F_{t+\dt}|\sigma_{t+\dt},\IVstep)\right) = E(\condFstep) = F_t,
\end{equation}
when $F_{t+\dt}$ is sampled according to Algorithm~\ref{alg:alg_cev}. The satisfaction of the martingale condition significantly improves the accuracy of our simulation algorithm (see numerical results illustrated in Figure~\ref{f:Islah_comparison}).
\subsection{Sampling conditional terminal forward price from a CEV distribution} \label{ss:cev_sample} \noindent
With availability of analytic distribution function, the slow root-finding step is often used to invert the underlying distribution function. This approach has been adopted by some existing SABR simulation methods (e.g., \citet{cai2017sabr}). Our SABR simulation algorithm would be completed when our CEV approximation of the forward price in Algorithm~\ref{alg:alg_cev} is paired with an efficient sampling algorithm of the CEV distribution. It is fortunate that we can use the algorithms of \citet{makarov2010exact} and \citet{kang2014simulation} for the exact simulation of the CEV distribution. Despite its simplicity and efficiency, this exact simulation scheme of the CEV distribution is not widely known in the literature. First, we present the relevant probability distributions that are related to the CEV distribution and several related analytic properties. 
\begin{definition}[Noncentral chi-squared random variable]
	Let $\NCX(\delta, r)$ be the NCX2 random variable with degree of freedom $\delta$ and noncentrality $r$. The PDF and CDF of $\NCX(\delta, r)$ are respectively given by
	\begin{equation}
		f_\NCX(x\,;\,\delta,r) = 
		\frac12\left(\frac{x}{r}\right)^{\alpha/2}
		I_\alpha(\sqrt{x\,r})\; e^{-(x+r)/2} \qtext{and}
		P_\NCX(x;\delta,r)=\int_{0}^{x} f_\NCX(t\,;\,\delta,r)\, \mathrm{d}t,
	\end{equation}
	where $\alpha = \delta/2 - 1$ and $I_\alpha(z)$ is the modified Bessel function of the first kind, 
	$$ I_\alpha(z) = \sum_{k=0}^{\infty} \frac{(z/2)^{\alpha + 2k}}{k! \,\Gamma(k+\alpha+1)}.
	$$
\end{definition}

\begin{remark}[Relation to the gamma distribution] \label{r:ncx2_gamma}
	The central chi-squared (corresponding to $r=0$) distribution is a gamma distribution. Let $\GAM(\alpha)$ be a gamma random variable with shape parameter $\alpha$ and unit scale parameter, the two random variables are related by
	$$ \NCX(\delta, 0) \sim 2\, \GAM(\delta/2).
	$$
	The PDF and CDF of $\GAM(\alpha)$ are respectively given by~\footnote{The CDF of $\GAM(\alpha)$, $P_{\GAM}(x;\alpha)$, is the scaled lower incomplete gamma function $\gamma(\alpha, x)$:
		$$ P_{\GAM}(x;\alpha) = \frac{\gamma(\alpha, x)}{\Gamma(\alpha)}  \qtext{where}
		\gamma(\alpha,x) = \int_{0}^{x}t^{\alpha-1}e^{-t}\,\mathrm{d}t.$$
	}
	\begin{equation}\label{eq:incomplete_gamma}
		f_{\GAM}(x; \alpha) = \frac1{\Gamma(\alpha)} x^{\alpha-1}e^{-x} \qtext{and}
		P_{\GAM}(x;\alpha) = \frac1{\Gamma(\alpha)} \int_{0}^{x}t^{\alpha-1}e^{-t}\,\mathrm{d}t.
	\end{equation}
\end{remark}

\begin{definition}[Shifted Poisson variate]
	The shifted Poisson (SP) variate with intensity $\lambda$ and shift parameter $\alpha$, $N\sim\SP(\lambda,\alpha)$, takes non-negative integer values with mass probability function:
	\begin{equation}\label{eq:SP}
		P_\SP (n;\lambda, \alpha) = \mathrm{Prob}(N=n) = \frac{1}{P_{\GAM}(\lambda;\alpha)} \frac{\lambda^{\alpha+n}\,e^{-\lambda}}{\Gamma(n+\alpha+1)}.
	\end{equation}
\end{definition}
\begin{remark}[Relation to the Poisson distribution]
	When the shift parameter $\alpha$ is zero, the SP distribution is reduced to the Poisson distribution with intensity $\lambda$, $\POIS(\lambda)$, with mass probability function:
	$$		P_\POIS (n;\lambda) = \frac{\lambda^{n}\,e^{-\lambda}}{n!}.
	$$	
\end{remark}
\begin{remark}
	The SP variate arises from the series expansion of $P_{\GAM}(x;\alpha)$~\citep[6.5.29]{abramowitz}:
	$$
	1 = \frac1{P_{\GAM}(\lambda;\alpha)} \frac{\gamma(\alpha, \lambda)}{\Gamma(\alpha)} = \frac1{P_{\GAM}(\lambda;\alpha)}\sum_{n=0}^{\infty}\frac{\lambda^{\alpha+n}\,e^{-\lambda}}{\Gamma(n+\alpha+1)} = \sum_{n=0}^{\infty} P_\SP(n;\lambda, \alpha).$$
\end{remark}

Next, we discuss the connection of the above defined distributions to the CEV distribution. It is well known that the transition density of $F_T$ is closely related to the NCX2 distribution~\citep{schroder1989comp,cox1996cev}. We define $z_T$ be a transformation of the CEV process $F_T$ (see Eq.~\eqref{eq:def_CEV}) by the function $z(\cdot)$:
\begin{equation}\label{eq:z_0}
	z_T := z(F_T)
	\qtext{and} z(y) := \frac{y^{2\betac}}{\betac^2\sigma_0^2 T}.
\end{equation}
The PDF of $z_T > 0$ is expressed in terms of the NCX2 distribution:
$$ 	f_{z_T}(z;z_0)\, \mathrm{d}z = \mathrm{Prob}(z_T \in [z, z+\mathrm{d}z]) = f_\NCX\left(z_0; \frac{1}{\betac}+2 , z\right) \mathrm{d}z\quad (z>0).$$
Since the location variable $z$ plays the role of the noncentrality parameter, the distribution of $z_T$ is not a NCX2 distribution. 
The complementary CDF of $F_T>0$ and the probability of absorption (probability mass at zero) are respectively given by
\begin{subequations} 
	\begin{equation}\label{eq:cev_CDF}
		\mathrm{Prob}(F_T > y) = \mathrm{Prob}(z_T > z(y)) = P_\NCX\left(z_0; \frac{1}{\betac}, z(y)\right),
	\end{equation}
	\begin{equation}\label{eq:cev_PDF}
			\mathrm{Prob}(F_T=0) = \mathrm{Prob}(z_T=0) = 1 - P_\NCX\left(z_0;\, \frac1{\betac},0\right) = 
		1 - P_{\GAM}\left(\frac{z_0}{2};\,\frac{1}{2\betac} \right).
	\end{equation}
\end{subequations}
The last equality in the second line comes from Remark~\ref{r:ncx2_gamma}. 

\begin{proposition}[Mixture gamma representation of the CEV distribution] \label{pr:proposition1}
	Under the CEV model with absorbing boundary condition, given $z_T>0$ ($F_t$ is not absorbed at the origin until time $T$), the transition from $z_0$ to $z_T$ can be sampled by an SP-mixture gamma distribution \citep[\S~3.3]{makarov2010exact}:
	\begin{equation}\label{eq:SP_mixture}
		 z_T \sim 2\, \GAM(N+1) \qtext{where} N\sim \SP\left(\frac{z_0}{2}, \frac{1}{2\betac}\right).
	\end{equation}
\end{proposition}
\begin{proof}
We employ the infinite series expansion of $I_\alpha(\sqrt{z_0\,z_T})$ to obtain 
\begin{align*}
	f_\NCX\left(z_0; \frac{1}{\betac}+2 , z_T\right) &= \frac12\left(\frac{z_0}{z_T}\right)^{\alpha/2}
	I_\alpha(\sqrt{z_0\,z_T})\; e^{-(z_T+z_0)/2} \quad \left(\alpha = \frac{1}{2\betac}\right)\\
	&= \frac12\left(\frac{z_0}{z_T}\right)^{\alpha/2} \sum_{k=0}^{\infty} \frac{(\sqrt{z_0\,z_T}/2)^{\alpha + 2k}}{k! \,\Gamma(k+\alpha+1)} \; e^{-(z_T+z_0)/2}\\
	&= \frac12 \sum_{k=0}^{\infty} \frac{(z_0/2)^{k+\alpha}e^{-z_0/2}}{\Gamma(k+\alpha+1)} \frac{(z_T/2)^{k}}{k!} \; e^{-z_T/2}\\
	&= P_{\GAM}\left(\frac{z_0}{2};\,\alpha \right) \sum_{k=0}^{\infty} P_\SP\left(k;\frac{z_0}{2},\alpha\right)\cdot \frac{1}{2}f_{\GAM}\left(\frac{z_T}{2};k+1\right),
\end{align*}
so that the transition density of $z_T$ conditional on $z_T>0$ is given by
$$
\frac{f_{z_T}(z_T;z_0)}{\mathrm{Prob}(z_T > 0)} = \sum_{k=0}^{\infty} P_\SP\left(k;\frac{z_0}{2},\alpha\right)\cdot \frac{1}{2}f_{\GAM}\left(\frac{z_T}{2};k+1\right).
$$
The right hand side is the composite probability density of the gamma variable, $2\,\GAM(k+1)$, where $k$ is sampled from the SP variable, $\SP\left(z_0/2, 1/2\betac\right)$. 
Therefore, $z_T>0$ can be sampled as $2\,\GAM(N+1)$.
\end{proof}
As proposed by \cite{kang2014simulation}, the SP variable, $N\sim \SP(\lambda, \alpha)$, can be sampled by a gamma-mixture Poisson variable:
\begin{equation}\label{eq:gamma_mixture_poisson}
	N \sim \POIS(\lambda-X)	\qtext{where} X\sim\GAM(\alpha) \qtext{conditional on} X<\lambda. 
\end{equation}
Combining Eqs.~\eqref{eq:SP_mixture} and \eqref{eq:gamma_mixture_poisson} leads to an exact simulation algorithm for sampling the CEV distribution. Moreover, Scheme~\eqref{eq:gamma_mixture_poisson} can be further optimized in the context of the CEV simulation. This is the reason we select `Algorithm 3' among the three algorithms in \citet{kang2014simulation}. As this advantage is not mentioned explicitly in \citet{kang2014simulation}, we state this feature as a remark:
\begin{remark}[\citet{kang2014simulation}'s `Algorithm 3' for the CEV distribution]
	Suppose simulating $N\sim \SP(\lambda,\alpha)$ is the final goal, $X < \lambda$ can be obtained by repeatedly sampling $X$ until $X < \lambda$. If $N$ is an intermediate variable for the CEV simulation in Scheme~\eqref{eq:SP_mixture}, however, the repeated sampling of $X$ is unnecessary since
	$$ \mathrm{Prob}(X \ge \lambda) = 1 - P_{\GAM}\left(\frac{z_0}{2};\,\frac{1}{2\betac} \right) = \mathrm{Prob}(F_T=0).
	$$
Therefore, we set $F_T=0$ when $X\ge \lambda$.
\end{remark}

In summary, an exact simulation of the CEV distribution consists of successive generation of elementary random variables, for which efficient numerical algorithms are available in standard numerical library. Finally, the simulation algorithm for $F_T\sim \cev(F_0,\sigma_0^2T)$ is summarized in Algorithm~\ref{alg:alg1}:

\vspace{1em}
\begin{algorithm}[H] 
	\setlength{\abovedisplayskip}{10pt}
	\setstretch{1}
	\caption{Exact simulation of the CEV distribution of $F_T$ as depicted in Eq.~\eqref{eq:def_CEV}} \label{alg:alg1}
	\KwIn{$F_0$, $\sigma_0$, $\beta$, $T$}
	\KwOut{$F_T$}
	$\alpha\gets \frac1{2\betac}$ and $z_0 \gets \frac{F_0^{2\betac}}{\betac^2\sigma_0^2 T}$\;
	$X \gets \GAM\left(\alpha\right)$\;
	\eIf{$X\ge z_0/2$} 
	{
		\KwRet{$F_T \gets 0$}\;
	}{
		$z_T \gets 2\,\GAM(\POIS(z_0/2-X) + 1)$\;
		\KwRet{$F_T \gets \left(\betac^2\sigma_0^2Tz_T \right)^{1/2\betac}$}\;
	}
\end{algorithm}
\vspace{1em}


The CEV sampling algorithm can be applied to the conditional distribution of $F_{t+\dt}$ in Eqs.~\eqref{eq:beta01} and \eqref{eq:Ft_cond_new}. We need to sample $z_{t+\dt} = z(F_{t+\dt})$ from $z_t = z(\condFstep)$, where the transformation $z(\cdot)$ is modified to become
$$ z(y) = \frac{y^{2\betac}}{\betac^2 \rhoc^2 \sigma_t^2 \dt \IVstep}.
$$
We summarize the simulation of $F_t$ in the context of the SABR model in Algorithm~\ref{alg:alg2}:

\vspace{1em}
\begin{algorithm}[H]
	\caption{Simulation of $F_T$ under the SABR model} \label{alg:alg2}
	\KwIn{$F_0$, $\sigma_0$, $\vov$, $\beta$, $\rho$, $T$, $\dt$}
	\KwOut{$F_T$}
	$\alpha\gets \frac1{2\betac}$, $t \gets 0$, $ F_t \gets F_0$, and $\sigma_t \gets \sigma_0$\;
	\Repeat{t = T}{
		$ \sigma_{t+\dt} \gets $ Eq.~\eqref{eq:sig_t} \;
		$ \IVstep \gets $ Algorithm~\ref{alg:alg1_IV} \;
		$ \condFstep \gets $ Eq.~\eqref{eq:Ft_cond_new} \;
		$z_t \gets \frac{\left(\condFstep\right)^{2\betac}}{\betac^2 \rho^2 \sigma_t^2 \dt \IVstep}$\;
		$X \gets \GAM\left(\alpha\right)$ \;
		\eIf{$X\ge z_t/2$} 
		{
			$F_{t+h} \gets 0$ \;
			\textbf{break} \;
		}{
			$z_{t+\dt} \gets 2\,\GAM(\POIS(z_t/2 - X) + 1)$\;
			$F_{t+\dt} \gets \left(\betac^2 \rho^2 \sigma_t^2 \dt \IVstep z_{t+\dt}\right)^{1/2\betac}$\;
		}
		$t \gets t+h$ \;
	}
	\KwRet{$F_T$} \;
\end{algorithm}
\vspace{1em}


\subsection{Comparison to Islah's approximation used in other simulation schemes}\label{ss:compare_F_T}\noindent
While the earlier SABR simulation methods use various approaches for sampling $\IVstep$ (Simulation step~1), almost all of them~\citep{chen2012low, cai2017sabr,leitao2017one,leitao2017multi, grzelak2019stochastic,cui2021efficient,kyriakou2023unified} are based on Islah's approximation for sampling $F_{t+\dt}$ (Simulation step~2).  We compare our approach to \citet{islah2009sabr-lmm}'s in this section. We state Islah's approximation as follows. For $0<\beta<1$, conditional on $\sigma_{t+\dt}$ and $\IVstep$ (and the filtration up to time $t$), the distribution of $F_{t+\dt}$ is approximated by
\begin{equation} \label{eq:Ft_islah}
	\mathrm{Prob}(F_{t+\dt} \ge y) \approx P_\NCX\left(z'_t; \frac{1 - \betac\rho^2}{\betac\rhoc^2}, z(y)\right),
\end{equation}
where
$$	
z'_t = \frac{1}{\betac^2\rhoc^2\sigma_t^2 \dt \IVstep}\left[F_t^{\betac} + \frac{\betac\rho}{\vov}(\sigma_{t+\dt} - \sigma_t)\right]^2 \qtext{and} z(y) = \frac{y^{2\betac}}{\betac^2\rhoc^2\sigma_t^2 \dt \IVstep}.
$$
For the rationales behind this approximation, see \citet[\S~2.2]{islah2009sabr-lmm}, \citet[Result~2.4]{chen2012low}, or \citet[Appendix~B.1]{cui2021efficient}.

Islah's approximation is compared with the probability under our CEV approximation in Eqs.~\eqref{eq:beta01} and \eqref{eq:Ft_cond_new}:
$$ \mathrm{Prob}(F_{t+\dt} \ge y) \approx P_\NCX\left(z_t; \frac{1}{\betac}, z(y)\right)
\qtext{with} z_t = z\left(\condFstep\right).
$$
The key difference between these two approximations is that Islah's approximation uses the modified degree of freedom, $(1 - \betac\rho^2)/(\betac\rhoc^2)$, in Eq.~\eqref{eq:Ft_islah}, and it is not equal to $1/\betac$ (except $\rho=0$). However, the transformation $z(\cdot)$ uses $\betac$ inconsistently. On the other hand, our CEV approximation uses $1/\betac$, which is the same degree of freedom for the CEV distribution. This means that $F_{t+\dt}$ under Islah's approximation is not a CEV distribution, which leads to two drawbacks discussed in the next two remarks. 

\begin{remark}
	As $\vovn\downarrow 0$, $F_{t+h}$ under Islah's approximation does not converge to the CEV distribution with $\beta$ if $\rho\neq 0$. This violates the limiting property that the SABR model should converge to the CEV model with $\beta$ as $\vov\downarrow 0$.
\end{remark}

\begin{remark}
	The martingale property of $F_t$ is not guaranteed under Islah's approximation: $F_t \neq E(F_{t+h}).$ In particular, when $\beta\downarrow 0$ ($\betac\uparrow 1$), we can show that $F_t$ is a supermartingale: $F_t \le E(F_{t+\dt})$. In the limit, $(1 - \betac\rho^2)/(\betac\rhoc^2)$, as well as $1/\betac$, approaches to 1, and the conditional expectation is given by
	$$ E\left(F_{t+\dt} \,|\, \sigma_{t+\dt}, \IVstep\right) = z^{-1}\left(z_t'\right) = \left|F_t + \frac{\rho}{\vov}(\sigma_{t+\dt} - \sigma_t)\right|.
	$$
	Taking expectation on both sides over the joint distribution of $(\sigma_{t+\dt}, \IVstep)$, we arrive at
	\begin{gather*}
	F_t = E\left(F_t + \frac{\rho}{\vov}(\sigma_{t+\dt} - \sigma_t)\right)
	\le E\left(\left|F_t + \frac{\rho}{\vov}(\sigma_{t+\dt} - \sigma_t)\right|\right)
	= E\left(E\left(F_{t+\dt} \,|\, \sigma_{t+\dt}, \IVstep\right)\right) = E(F_{t+\dt}).
\end{gather*}
\end{remark}

In order to preserve the martingale property, \citet{leitao2017one,leitao2017multi} artificially add an additive adjustment, commonly called the martingale correction, to the simulated samples of $F_{t+\dt}$. This ad-hoc martingale correction is inferior when compared with our approach where the martingale condition is naturally observed.

Apart from the failure of the martingale property in Islah's approximation, a fast and exact algorithm for sampling $F_{t+\dt}$ according to Eq.~\eqref{eq:Ft_islah} has not been available in the literature. 
The methods presented so far are either slow (but accurate) or inaccurate (but fast).
\citet{chen2012low} (partially) and \citet{cai2017sabr} adopted numerical root finding of $P_\NCX(\,\cdot\,, \,\cdot\,, z_K) = U$ for a uniform random number $U$, which belongs to the former type. This approach is slow due to the cumbersome evaluations of $P_\NCX$ despite using some good initial guess and efficient calculation algorithm. \citet{grzelak2019stochastic} speeded up sampling by using the stochastic collocation method, which is an optimal caching of the distribution function.
\citet{chen2012low} (partially) and \citet{cui2021efficient} adopted the moment-matched quadratic Gaussian approximation to Eq.~\eqref{eq:Ft_islah}, which is proven to be effective in the Heston simulation~\citep{andersen2008simple}. While being fast, it becomes inaccurate for a larger time step as it only uses the first two moments. This approach belongs to the latter type. 

It is interesting to note that the exact CEV sampling algorithm in Section~\ref{ss:cev_sample} can also be modified to sample $F_{t+\dt}$ under Islah's approximation, though it is not recommended due to the inherent drawbacks of Islah's approximation. We outline the simulation algorithm in \ref{apdx:islah_cev}. With the use of the exact CEV sampling procedure, we numerically illustrate the supermartingale property of $F_t$ in Section~\ref{s:num}.

\section{Numerical Results} \label{s:num}\noindent
We demonstrate numerical accuracy and effectiveness of our simulation algorithm with an extensive set of numerical examples. In Subsection~\ref{ss:error_separate}, we examine separately the errors caused by the two steps of simulation approximation: (i) SLN approximation of the average variance, (ii) CEV approximation of the forward price distribution. We then compare accuracy and speed tradeoff with analytic approximation formula pf European option prices in Subsection~\ref{ss:accuracy} and other simulation schemes in Subsection~\ref{ss:error}. In the last subsection, we show the competitive edge of accuracy performance of our CEV approximation of the forward price distribution with that of the Islah's approximation.

For easier comparison to other existing methods, we select the sets from the previous literature: the first two from \citet{antonov2012adv} and the rest from \citet{cai2017sabr}. Following the literature~\citep{chen2012low,cai2017sabr}, we price European options. Although the versatility of our simulation scheme goes beyond prcing European options, it offers a good testing case. Pricing European options also benefits from the availability of highly accurate benchmark option prices obtained from using the finite difference method (FDM) \citep{cai2017sabr}. 

In our numerical tests, unless otherwise specified, we used $N=10^5$ paths for each simulation, and repeated $m=50$ times. Out of $m$ prices, we computed the average bias (used the FDM price as the benchmark) and standard deviation. Our algorithm is implemented in MATLAB R2023b on a laptop with a 13th Gen Intel Core i7-1360P 2.20 GHz processor. In Section~\ref{ss:error}, to demonstrate the speed of algorithm, we compare the CPU time of our method with those reported in \citet{cai2017sabr}. \citet{cai2017sabr} used Matlab 7 on an Intel Core2 Q9400 2.66GHZ processor. We obtained the computer codes of \cite{cai2017sabr} and reran the codes again using our computer. The CPU times of their algorithm have been updated in Table~\ref{t:case3} and Figure~\ref{f:RMS}, where we examine the errors versus CPU time tradeoff. As shown in the comparison, our algorithm is seen to remain faster than the existing algorithms by at least one order of magnitude.

\subsection{Examination of numerical errors caused by separate simulation steps}\label{ss:error_separate}\noindent
Recall that there are two separate simulation steps that may contribute to the overall errors in the full simulation scheme. The first step involves the SLN approximation of the average variance. The second step involves the CEV approximation of the forward price. The second step gives rise to errors from the operator splitting procedure and frozen coefficient approximation. The simulation of the CEV distribution via the SP mixture Gamma variate is exact, so this does not lead to simulation errors. We manage to isolate the errors that are separated into contributions from the SLN approximation and frozen coefficient approximation by considering accuracy performance of several degenerate cases of the SABR model. The accuracy analysis of the overall errors in the full simulation scheme in the general SABR model is discussed in Subsections~\ref{ss:accuracy} and \ref{ss:error}.

First, we consider the lognormal SABR model, where $\beta=1$. The corresponding conditional forward price $F_T$ follows the lognormal distribution (see Eq.~\eqref{eq:beta1}), so one does not need to introduce operator splitting and frozen coefficient procedures to obtain an approximation of conditional $F_T$. The lognormal SABR model serves to isolate the error contributed by the SLN approximation only. We performed accuracy analysis of European option prices computed using our simulation scheme with the SLN approximation and the benchmark numerical results obtained from the finite difference method (FDM). We compute the relative error using the formula:
$$\text{relative error}=\frac{\text{simulation calculation}-\text{FDM}}{\text{FDM}}.$$
In our calculations, we use the following parameter values: $\beta=1$, $T=1$, $\vov=0.2$, $\rho=-0.75$ and $h=1$ for the base case, and perturb the parameter values to $\rho=-0.5$, $\rho=-0.25$, $\vov=0.4$, $\vov=0.8$, $h=0.5$ and $h=0.25$ as well. The other parameter values include $\sigma_0=0.2$ and $F_0=K=1$. The relative errors of the simulation calculations are listed in Table~\ref{t:SLN_error}.

\begin{table}[!h]
	\centering
	\begin{tabular}{|c|c|c|c|c|c|c|c|}
		\hline
		$\rho$ & -0.75 & -0.75 & -0.75 & -0.5 & -0.25 & -0.75 & -0.75 \\ \hline
		$\vov$ & 0.2 & 0.2 & 0.2 & 0.2 & 0.2 & 0.4 & 0.6 \\ \hline
		$h$ & 1 & 0.5 & 0.25 & 1 & 1 & 1 & 1 \\
		\hline
		$\text{Rel. Err.}$ & 0.00353\% & 0.00489\% & 0.0110\% & 0.00700\% & 0.00275\% & 0.00808\% & 0.0198\%\\
		\hline
	\end{tabular}
	\caption{The relative errors contributed by the SLN approximation in the calculations of European call option prices for the lognormal SABR model ($\beta=1$). The FDM value for $\rho=-0.75$, $\vov=0.2$ is 0.07910, $\rho=-0.5$, $\vov=0.2$ is 0.07942, $\rho=-0.25$, $\vov=0.2$ is 0.07969, $\rho=-0.75$, $\vov=0.4$ is 0.07860 and $\rho=-0.75$, $\vov=0.6$ is 0.07811.}
	\label{t:SLN_error}
\end{table}

We observe that the errors contributed by the SLN approximation for the calculations of European option prices for the lognormal SABR model ($\beta=1$) are exceedingly small, only a very small fraction of one percentage point. In most cases, the agreement of the simulation result and FDM value is up to 3 significant figures. The errors increase with higher value of $\vov$, but they remain to be very small. The decrease of time step $h$ is irrelevant in decreasing the errors. Indeed, errors increase due to roundoffs with more calculation due to an increase of the number of time steps. In conclusion, the SLN approximation of the average variance is highly accurate and contributes almost insignificant errors.

Next, we consider the SABR models with $\rho=1$, $\rho=0.75$ and $\rho=0$ while $0<\beta<1$. The full correlation case ($\rho=1$) exhibits errors in calculating the European option prices contributed by the SLN approximation (should be exceedingly small from the above discussion) and frozen coefficient approximation, but no errors contributed from the operator splitting procedures. The conditional forward price of the zero correlation ($\rho=0$) SABR model follows the CEV distribution exactly, so there are no errors contributed by the operator splitting and frozen coefficient procedures. The errors arise only from the SLN approximation and simulation of the CEV distribution via the SP mixture Gamma variate. For $\rho=0.75$, it should have all three error sources: SLN approximation, operator splitting and frozen coefficient approximation. Table~\ref{t:rho_changes} lists the relative errors of the simulation calculations, which are benchmarked with the FDM results. Parameter values for $\sigma_0$, $F_0$, $K$ and $T$ are the same as those used in Table~\ref{t:SLN_error}.
\begin{table}[!h]
	\centering
	\begin{tabular}{|c|c|c|c|c|c|}
		\hline
		\multirow{2}{*}{$\rho=1$}& $\vov=0.2$ & $\vov=0.2$ & $\vov=0.2$ & $ \vov=0.4$ & $\vov=0.8$\\
	 &$\beta=0.4$ & $\beta=0.6$ & $\beta=0.8$ & $\beta=0.8$ & $\beta=0.8$ \\
		\hline
		$h=1$ & 0.518\%  & 0.348\% &0.164\% & 0.404 & 0.746\%  \\
		$h=0.5$& 0.418\% & 0.225\% & 0.0683\% & 0.199 & 0.350\% \\
		$h=0.25$ & 0.244\% & 0.119\% & 0.0299\% & 0.0947 & 0.224\%\\
		\hline\hline
		\multirow{2}{*}{$\rho=0.75$}& $\vov=0.2$ & $\vov=0.2$ & $\vov=0.2$ & $ \vov=0.4$ & $\vov=0.8$\\
		&$\beta=0.4$ & $\beta=0.6$ & $\beta=0.8$ & $\beta=0.8$ & $\beta=0.8$ \\
		\hline
		$h=1$ & 0.415\% & 0.306\% & 0.125\% & 0.333\% & 0.421\%  \\
		$h=0.5$& 0.237\% & 0.130\% & 0.0647\% & 0.0773\% & 0.301\% \\
		$h=0.25$ & 0.0399\% & 0.0287\% & 0.0242\% & 0.0597\% & 0.215\%\\
		\hline\hline
		\multirow{2}{*}{$\rho=0$}& $\vov=0.2$ & $\vov=0.2$ & $\vov=0.2$ & $ \vov=0.4$ & $\vov=0.8$\\
		&$\beta=0.4$ & $\beta=0.6$ & $\beta=0.8$ & $\beta=0.8$ & $\beta=0.8$ \\
		\hline
		$h=1$ & -0.0562\% & -0.00574\% & 0.0704\% & 0.0257\% & 0.123\%  \\
		$h=0.5$& -0.0138\% & 0.0151\% & -0.0537\% & 0.0803\% & 0.0292\% \\
		$h=0.25$ & -0.0526\% & 0.0123\% & -0.0162\% & 0.0616\% & 0.0758\%\\
		\hline
	\end{tabular}
	\caption{The relative errors in calculating the European call option prices for the full correlation ($\rho=1$), $\rho=0.75$ and zero correlation ($\rho=0$) SABR models. The FDM value for $\rho=1$: $\vov=0.2\,,\beta=0.4$ is 0.07989, $\vov=0.2,\,\beta=0.6$ is 0.08002, $\vov=0.2,\,\beta=0.8$ is 0.08017, $\vov=0.4,\,\beta=0.8$ is 0.08044 and $\vov=0.8,\,\beta=0.8$ is 0.08043. The FDM values for $\rho=0.75$ are 0.07998, 0.08008, 0.08018, 0.08083 and 0.08276, respectively. The FDM values for $\rho=0$ are 0.07996, 0.07994, 0.07992, 0.08068 and 0.08355, respectively.}
	\label{t:rho_changes}
\end{table}

For the full correlation ($\rho=1$) SABR model, the main source of errors comes from the frozen coefficient approximation. In all cases of different values of $\vov$ and $\beta$, the errors show decrease in magnitude with respect to the time step $h$. The frozen coefficient approximation is sensitive to the size of the time step $h$, smaller approximation error with smaller time step. Also, the errors show increase in magnitude with respect to $\betac=1-\beta$ and $\vov$. Such trends hold similarly for $\rho=0.75$, whose errors are from the three main sources: SLN approximation, operator splitting and frozen coefficient approximation. It is seen that the errors from operator splitting and frozen coefficient approximation may counteract with each other, leading to a smaller error when compared with the full correlation ($\rho=1$). On the other hand, for zero correlation ($\rho=0$) SABR model, the sources of errors come from the SLN approximation and exact simulation of the CEV distribution, both are exceedingly small (well below one percentage point). There is no pattern of dependence on $\vov$, $\beta$ and $h$ since errors are mostly from roundoff.

\subsection{Comparison to analytic approximations} \label{ss:accuracy}\noindent
\begin{table}[!h]
	\centering
	\begin{tabular}{c|ccccccc}
		\hline
		{Cases} & $F_0$ & $\sigma_0$ & $\vov$ & $\rho$ & $\beta$ & $T$ & $K$\\
		\hline
		{Case \uppercase\expandafter{\romannumeral1}} & 1 & 0.25 & 0.3 & $-0.8$ & 0.3 & 10 & [0.2,\;2.0]\\
		{Case \uppercase\expandafter{\romannumeral2}} & 1 & 0.25 & 0.3 & $-0.5$ & 0.6 & 10 & [0.2,\;2.0]\\
		{Case \uppercase\expandafter{\romannumeral3}} & 0.05 & 0.4 & 0.6 & 0.0 & 0.3 & 1 & [0.02,\;0.10]\\
		{Case \uppercase\expandafter{\romannumeral4}} & 1.1 & 0.4 & 0.8 & $-0.3$ & 0.3 & 4 & 1.1\\ 
		{Case \uppercase\expandafter{\romannumeral5}} & 1.1 & 0.3 & 0.5 & $-0.8$ & 0.4 & [1,\;10] & 1.1\\ \hline
	\end{tabular}
	\caption{Parameter sets used in our numerical experiments.}
	\label{t:param}
\end{table}
Table~\ref{t:param} shows the five parameter sets used in our numerical tests. In this subsection, we test the accuracy of our algorithm (Algorithm~\ref{alg:alg2}) using Cases~\uppercase\expandafter{\romannumeral1} and \uppercase\expandafter{\romannumeral2} in Table~\ref{t:param}. Tables~\ref{t:case1} and \ref{t:case2} show the numerical results for the two sets.
The two examples are previously tested in \citet{antonov2012adv} to demonstrate the performance of the cutting-edge analytic approximations for European vanilla options, map to the zero correlation (ZC Map, and hybrid map to the zero correlation (Hyb ZC Map), against the widely used \citet{hagan2002sabr}'s implied volatility formula.\footnote{Since these approximation formulas are all expressed in terms of implied volatilities, we use the Black-Scholes formula to obtain the call option price. \citet{antonov2012adv} use MC method for the benchmark option price. Instead, we use FDM price for the benchmark, which is more precise.} Therefore, we also display those option prices for comparison. 

Tables~\ref{t:case1} and \ref{t:case2} demonstrate that our simulation scheme is highly accurate. As we decrease the time step $h$ from 1 to $1/16$, the option price converges to its true price. Even for a large time step ($h=1$), the option prices are accurate almost up to three decimal points. Such bias level is comparable to the advanced analytic methods. Overall speaking, we note that the bias in Table~\ref{t:case2} (Case~II) is lower than that in Table~\ref{t:case1} (Case~I). This is probably due to the observation that the geometric BM approximation for the conditional mean $\condF$ in Eq.~\eqref{eq:cev_gbm} holds better as $\beta=0.6$ (Case~II) is closer to $\beta=1$ (lognormal SABR model) than $\beta=0.3$ (Case~I).

\begin{table}[!h]
	\centering
	\begin{tabular}{|c||c|c|c|c|c|c|c|} \hline
$K/F_0$ & 0.2 & 0.4 & 0.8 & 1.0 & 1.2 & 1.6 & 2.0 \\ \hline\hline
\multicolumn{8}{|c|}{Exact option price from FDM} \\ \hline
FDM & 0.84255 & 0.68906 & 0.40646 & 0.28502 & 0.18304 & 0.05343 & 0.01096 \\ \hline\hline
\multicolumn{8}{|c|}{Bias of our simulation method\; $(\times 10^{-3})$} \\ \hline
$\dt=1$ & -1.22 & -1.49 & -0.37 & 0.49 & 1.28 & 1.72 & 1.32 \\
$\dt=1/4$ & -0.46 & -0.24 & 0.22 & 0.42 & 0.56 & 0.56 & 0.48 \\
$\dt=1/16$ & -0.34 & -0.20 & 0.00 & 0.05 & 0.11 & 0.10 & 0.10 \\ \hline
\multicolumn{8}{|c|}{Standard deviation of our simulation method\; $(\times 10^{-3})$} \\ \hline
$\dt=1$ & 1.97 & 1.83 & 1.50 & 1.31 & 1.08 & 0.63 & 0.38 \\
$\dt=1/4$ & 1.96 & 1.73 & 1.29 & 1.08 & 0.91 & 0.61 & 0.41 \\
$\dt=1/16$ & 1.89 & 1.75 & 1.44 & 1.28 & 1.06 & 0.53 & 0.22 \\ \hline\hline
\multicolumn{8}{|c|}{Bias of various analytic approximation methods\; $(\times 10^{-3})$} \\ \hline
Hagan & 22.35 & 23.64 & 17.94 & 13.81 & 9.38 & 2.54 & 0.82 \\
ZC Map & 0.37 & 0.51 & 1.39 & 2.29 & 3.20 & 4.02 & 2.66 \\
Hyb ZC Map & 4.64 & 5.81 & 4.07 & 2.29 & 0.43 & -1.26 & -0.30 \\ \hline
	\end{tabular}
	\caption{The European option prices from our simulation methods for Case \uppercase\expandafter{\romannumeral1}, compared with those from various analytic approximation methods (Hagan, ZC Map, and Hyb ZC Map) reported in \citet{antonov2012adv}. The errors are measured against the benchmark option prices from the finite difference method (FDM). We used $N=100,000$ paths in each simulation run and repeated $m=50$ times. The average computation time is 0.8, 1.6, and 14.5 seconds for $h=1$, $1/4$, and $1/16$, respectively.}
	\label{t:case1}
\end{table}

\begin{table}[!h]
	\centering
	\begin{tabular}{|c||c|c|c|c|c|c|c|} \hline
$K/F_0$ & 0.2 & 0.4 & 0.8 & 1.0 & 1.2 & 1.6 & 2 \\ \hline\hline
\multicolumn{8}{|c|}{Exact option price from FDM} \\ \hline
FDM & 0.82886 & 0.66959 & 0.39772 & 0.29118 & 0.20690 & 0.10018 & 0.05014 \\ \hline\hline
\multicolumn{8}{|c|}{Bias of our simulation method \; $(\times 10^{-3})$} \\ \hline
$\dt=1$ & -0.14 & -0.30 & -0.42 & -0.43 & -0.43 & -0.40 & -0.30 \\
$\dt=1/4$ & 0.45 & 0.37 & 0.27 & 0.20 & 0.10 & -0.02 & 0.00 \\
$\dt=1/16$ & 0.01 & -0.01 & 0.02 & 0.04 & 0.03 & 0.00 & -0.03 \\ \hline
\multicolumn{8}{|c|}{Stdev of our simulation method \; $(\times 10^{-3})$} \\ \hline
$\dt=1$ & 2.23 & 2.09 & 1.78 & 1.65 & 1.51 & 1.20 & 0.93 \\
$\dt=1/4$ & 2.21 & 2.10 & 1.85 & 1.70 & 1.51 & 1.14 & 0.88 \\
$\dt=1/16$ & 2.46 & 2.32 & 2.01 & 1.79 & 1.58 & 1.22 & 0.97 \\ \hline\hline
\multicolumn{8}{|c|}{Bias of various analytic approximation methods \; $(\times 10^{-3})$} \\ \hline
Hagan & 11.65 & 15.94 & 16.69 & 14.67 & 12.18 & 8.56 & 6.88 \\
ZC Map & -1.56 & -1.57 & 0.56 & 2.37 & 4.02 & 5.77 & 5.33 \\
Hyb ZC Map & 3.07 & 4.53 & 3.86 & 2.37 & 1.00 & -0.10 & 0.52 \\ \hline
	\end{tabular}
	\caption{The European option prices from our simulation methods for Case \uppercase\expandafter{\romannumeral2}, compared with those from various analytic approximation methods (Hagan, ZC Map, and Hyb ZC Map) reported in \citet{antonov2012adv}. The errors are measured against the benchmark option prices from the finite difference method (FDM). We used $N=100,000$ paths in each simulation run and repeated $m=50$ times.}
	\label{t:case2}
\end{table}

\subsection{Accuracy and speed trade-off} \label{ss:error}\noindent
In this subsection, we demonstrate the computational efficiency of our method by comparing the accuracy-speed trade-off with the earlier simulation methods, such as the Euler, low-bias~\citep{chen2012low}, and piecewise semi-exact (PSE)~\citep{cai2017sabr} schemes. We use Cases~III and IV to achieve better comparison since these cases have been tested by \citet{cai2017sabr}. We take advantage of the extensive records of the CPU time reported therein.

\begin{table}[!h]
	\centering		
	\begin{tabular}{|c||c|c|c|c|c|c||c|} \hline
		$K/F_0$ & 0.4 & 0.8 & 1 & 1.2 & 1.6 & 2 & Time (s) \\ \hline\hline
		\multicolumn{8}{|c|}{Exact option price from FDM} \\ \hline
		FDM & 0.04559 & 0.04141 & 0.03942 & 0.03750 & 0.03390 & 0.03061 &   \\ \hline\hline
		\multicolumn{8}{|c|}{Error $(\times 10^{-3})$ and CPU time of our simulation method} \\ \hline
		$\dt=1$ & 0.00 & 0.00 & 0.00 & 0.00 & -0.01 & -0.01 & 0.03 \\ \hline\hline
		\multicolumn{8}{|c|}{Bias $(\times 10^{-3})$ and CPU time of the Euler scheme reported in \cite{cai2017sabr}} \\ \hline
		$\dt=1/400$ & 1.6 & 1.5 & 1.5 & 1.4 & 1.3 & 1.2 & 4.12 \\ \hline
		$\dt=1/800$ & 0.7 & 0.6 & 0.5 & 0.5 & 0.4 & 0.3 & 8.26 \\ \hline
		$\dt=1/1600$ & -0.3 & -0.3 & -0.3 & -0.3 & -0.3 & -0.3 & 16.2 \\ \hline\hline
		\multicolumn{8}{|c|}{Bias $(\times 10^{-3})$ and CPU time of the low-bias scheme reported in \cite{cai2017sabr}} \\ \hline
		$\dt=1/4$ & 0.5 & 0.5 & 0.5 & 0.4 & 0.4 & 0.4 & 6.54 \\ \hline
		$\dt=1/8$ & 0.4 & 0.4 & 0.4 & 0.3 & 0.3 & 0.2 & 14.7 \\ \hline\hline
		\multicolumn{8}{|c|}{Bias $(\times 10^{-3})$ and CPU time of the PSE scheme reported in \cite{cai2017sabr}} \\ \hline
		$\dt=1$ & 0.1 & 0.2 & 0.2 & 0.2 & 0.2 & 0.2 & 8.19 \\ \hline
	\end{tabular}
	\caption{The European option prices from our simulation methods for Case \uppercase\expandafter{\romannumeral3}, compared with other MC methods: Euler, low-bias~\citep{chen2012low}, and PSE~\citep{cai2017sabr} schemes. We used the price and computation time reported in \citet[Table~7]{cai2017sabr}. The biases are measured against the benchmark option prices from the finite difference method (FDM). All simulation methods used $N=100,000$ paths in each simulation run and repeated $m=50$ times.}
	\label{t:case3}
\end{table}

The results for Case~III in Table~\ref{t:case3} show that our method is at least one hundred times faster than the competitors while the bias is much lower. Under $\rho=0$ in this case, both our CEV approximation and Islah's approximation become exact, and the source of bias is the accuracy of sampling $\IVstep$. The near-zero error of our method validates extremely high accuracy of our SLN approximation for $\IVstep$. For the Euler and low-bias schemes, it is necessary to decrease the time step $h$ for more accurate sampling of $\IVstep$, resulting with slower computation. Although the PSE scheme exhibits good accuracy for single step ($h=1$), it requires prohibitive computation for numerical inverse transform for sampling $\IVstep$. 

\begin{table}[!h]
	\centering
	\begin{tabular}{|r|c|c|r|} \hline
		$N$ \qquad &{$\dt$}&{RMS error ($\times 10^{-3}$)}&{Time (s)}\\ \hline
		160,000 & 1 & 3.27 & 0.53\\
		320,000 & 1/2 & 1.94 & 2.27\\
		640,000 & 1/4 & 1.21 & 9.66\\
		1,280,000 & 1/8 & 0.86 & 41.35\\
		2,560,000 & 1/16 & 0.59 & 279.53\\
		\hline 
	\end{tabular}
	\caption{The root mean square (RMS) error and CPU time from Case~IV when the number of paths $N$ is doubled and time step $\dt$ is halved.}
	\label{t:rms_cpu}
\end{table}

Next, we perform CPU time versus root-mean-square (RMS) error analysis on Case~IV. The RMS error is defined by
\begin{equation}\label{eq:rms}
	\text{RMS error}=\sqrt{\text{Stdev}^2+\text{Bias}^2}.
\end{equation}
Table~\ref{t:rms_cpu} presents the RMS error and CPU time for decreasing time step $\dt$ and increasing number of paths $N$. Similar to \citet[Table~9]{cai2017sabr}, we set $\dt \propto 1/N$. In each successive row, we double $N$ and halve $\dt$ to generate data points for the plots of RMS versus simulation time.

\begin{figure}[H]
	\centering
	\includegraphics[width=0.6\linewidth]{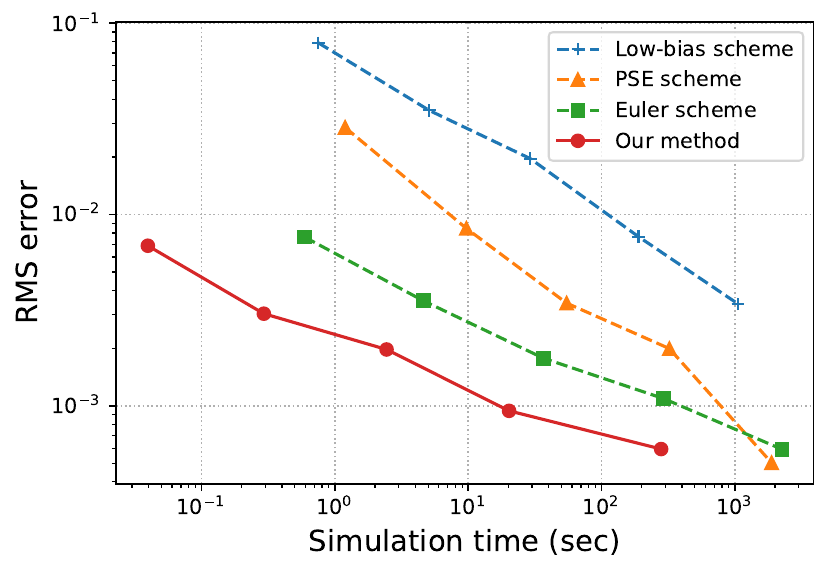}
	\caption{RMS error as a function of CPU time for simulation methods. The data points for our simulation scheme are from Table~\ref{t:rms_cpu} while those for other methods are from \citet[Table~9 and Figure~6]{cai2017sabr}.}
	\label{f:RMS}
\end{figure}

Based on the data points from Table~\ref{t:rms_cpu}, Figure~\ref{f:RMS} shows the plots of the trade-off between CPU time and RMS error. The figure explicitly demonstrates that, while the decay is slightly slower than the competing schemes, our simulation method achieves the same degree of RMS error using much less computation time, only about one hundredth of the Euler scheme. Thanks to the high efficiency for small time step $\dt$, our algorithm is also effective for evaluating path-dependent options with frequent discrete monitoring.

\subsection{Comparison with Islah's approximation} \label{ss:islah}\noindent
In this last numerical experiment using Case~V, we demonstrate superiority of our CEV approximation of conditional $F_T$ in Section~\ref{ss:our} to Islah's approximation which has been commonly adopted in other simulation schemes. The parameter set is also taken from \citet[Figure~4]{cai2017sabr}. We price the at-the-money European option for maturity, $T=1,2,\ldots, 10$ with time step $h=1/2$ or 1. On one hand, our simulation scheme consists of (i) SLN sampling of $\IVstep$ and (ii) CEV model sampling of conditional $F_T$. On the other hand, we replace the second step with Islah's approximation (Eq.~\eqref{eq:Ft_islah}).

Under Islah's approximation, a power transformation of $F_T$ observes a CEV distribution; see Eq.~\eqref{eq:CEV_islah}. Therefore, \citet{kang2014simulation}'s algorithm also provides a better way to simulate Islah's approach. Accordingly, we use Algorithm~\ref{alg:alg1} with $\beta$, $F_0$, and $\sigma_0^2 T$ replaced by $\beta'$, $\bar{F}_T'$, and $\rhoc^2\sigma_0^2 T \IV$ in Eq.~\eqref{eq:CEV_islah}, respectively, to sample $\left((\betac'/\betac) F_T^{\betac}\right)^{1/\betac'}$. Then, $F_T$ is finally obtained.  

\begin{figure}[h]
	\centering
	\subfigure[\small{Average terminal price $F_T$}]{
		\label{f:a}
		\includegraphics[width=0.47\linewidth]{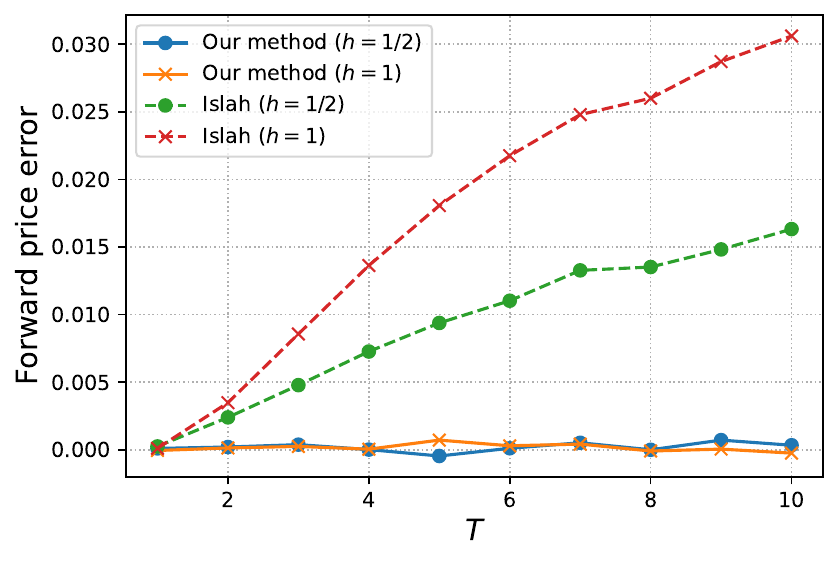} 
	}
	\subfigure[\small{European call option price}]{
		\label{f:b}
		\includegraphics[width=0.47\linewidth]{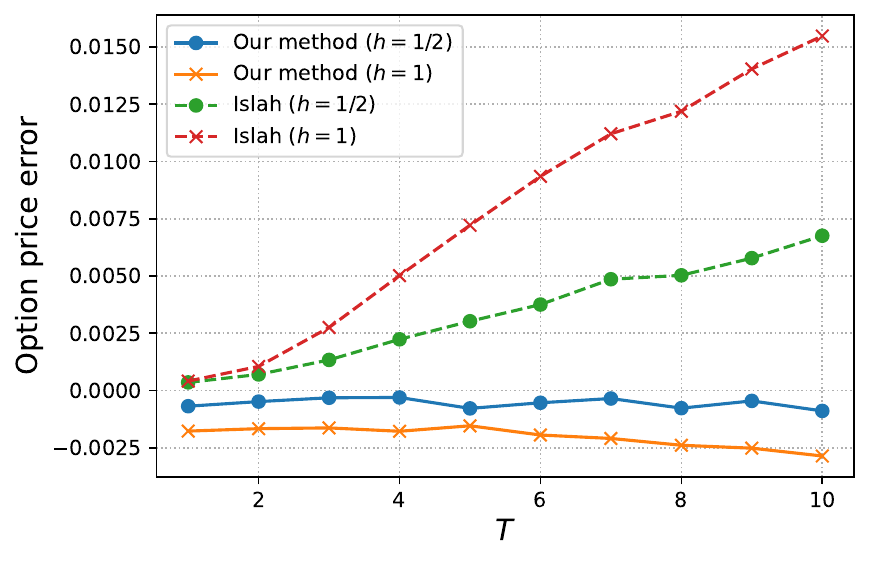} 
	}
	\caption{Comparison of Islah's approximation and our conditional distribution of $F_T$}
	\label{f:Islah_comparison}
\end{figure}

Figure~\ref{f:Islah_comparison} compares the results from the two approaches: our conditional distribution of $F_T$ and Islah's approximation. Figure~\ref{f:a} shows the error of the terminal price $E(F_T)$, which should be equal to $F_0=1.1$ in theory. As expected, the results show that our approach preserves the martingale property of the forward price fairly well for all time-to-maturity $T$. That is, $E(F_T)=F_0$ holds with very high accuracy regardless of $\dt$. In Islah's approach, however, $E(F_T)$ deviates from $F_0$. In particular, the deviation from $F_0$ accumulates as the time-to-maturity $T$ becomes longer. The time step $\dt$ has to be smaller to reduce the deviation. 

Figure~\ref{f:b} shows the error of the ATM option price from its true value computed with the FDM. Our approach is more accurate than Islah's approach. Note that the higher option price error in Islah's approach is due to failure of the martingale property. As the delta of ATM option is about 0.5, the option price error is approximately half of the forward price error. Conversely, the martingale preservation of our CEV approximation indeed improves European call option pricing. The option price error in our approach is mostly attributed to the error in $\IVstep$ distribution. The underpricing of option seems to be related to that the ex-kurtosis of $\IVstep$ sampled from the SLN approach, which is slightly smaller than that of true distribution (see Figure~\ref{f:moments}). The error is reduced when a smaller $h$ is used.

From Figure \ref{f:b}, we observe that under the same approach to sample $\IV$, and same size of time step $\dt$, the results using our CEV distribution of $F_T$ are very close to the benchmark FDM results. However, the results of Islah's approximation have larger numerical bias as $T$ gets longer. These results illustrate that our CEV approximation of the conditional $F_T$ is superior to Islah's approximation.
  
\section{Conclusion} \label{s:conc} \noindent
The stochastic-alpha-beta-rho (SABR) model proposed by \citet{hagan2002sabr} has been widely adopted for pricing derivatives under the stochastic volatility framework. The corresponding simulation methods have been studied extensively. The naive time-discretization schemes, such as the Euler or Milstein schemes, suffer from non-negligible bias caused by the truncation of negative forward values even when a small time step is used.
Although researchers proposed various simulation schemes that overcome the limitation of the time-discretization schemes, the schemes have their own limitations. In sampling conditional average variance, the simulation schemes are riddled with either heavy computation~\citep{cai2017sabr,leitao2017one} or inaccurate approximation over a large simulation step~\citep{chen2012low}. In sampling conditional forward price, almost all simulation algorithms adopt \citet{islah2009sabr-lmm}'s NCX2 approximation that neither preserves martingale property nor provides efficient sampling algorithm.

This paper present a series of innovations on the SABR simulation. Firstly, we derive the first four moments of the conditional average variance analytically, and sample the conditional average variance via an SLN random variable fitted to the first three moments. Our sampling procedure is accurate for reasonably large time-to-maturity and fast in computation in comparison to the numerical inverse transform methods. Secondly, we propose a martingale-preserving CEV approximation for the conditional forward price that preserves the martingale property. Lastly, we sample the conditional forward price efficiently with the exact CEV sampling algorithm of \citet{makarov2010exact} and \citet{kang2014simulation}. Numerical results show that our methods are highly efficient and accurate even under challenging parameter values, like large time-to-maturity $T$, when compared with other simulation methods.

\section*{Acknowledgment} \noindent
The works of Lilian Hu and Yue Kuen Kwok are supported by the Guangzhou-HKUST (GZ) Joint Funding Program (Number 2024A03J0630). The work of Yue Kuen Kwok is also supported by Project 2023CX10X1.

\appendix
\section{Conditional moments of $\IVstep$} \label{apdx:cond_mom} \noindent
The conditional average variance $\IVstep(\hat{Z})$ is closely related to the exponential functional of Brownian motion (BM), defined by
$$ A_t := \int_{0}^t e^{2 Z_s} \,\mathrm{d}s,
$$
where $Z_s$ is a standard BM. The derivation of Proposition~\ref{p:cond_mom} is based on the analytic conditional moments of the exponential functional of BM~\citep[(5.4)]{matsuyor2005exp1}:
$$ E\left(\left(A_t\right)^n | Z_t=x\right) = \frac{e^{nx}}{n! t} \int_x^\infty
b e^{(x^2-b^2)/2t} \left( \cosh b - \cosh x\right)^n \mathrm{d}b.
$$
The conditional average variance $\IVstep(\hat{Z})$ can be expressed by $A_t$ conditional on the terminal value of the BM, $Z_t$:
$$  \IVstep(\hat{Z}) = \frac{1}{\vovn^2} A_{\vovn^2} \,\Big|_{Z_{\vovn^2} = \vovn \hat{Z}}\, \quad \left(\vovn = \vov\sqrt{h}\right).
$$
Accordingly, the conditional moments of $\IVstep$ admit the following integral form:
\begin{align*}
\mu'_k &= E\left(\left(\IVstep\right)^k \right) = \frac{1}{\vovn^{2k}} E\left[\left(A_{\vovn^2}\right)^k \,\big|_{Z_{\vovn^2} = \vovn \hat{Z}}\right] \\	
	&= \frac{e^{k\vovn \hat{Z}}}{k!\,\vovn^{2k}} \int_{\hat{Z}}^\infty
	s e^{(\hat{Z}^2-s^2)/2} \left( \cosh\vovn s - \cosh\vovn \hat{Z} \right)^k \mathrm{d}s \\
	&= \frac{e^{k\vovn \hat{Z}}}{(k-1)!\,\vovn^{2k-1}} \int_{\hat{Z}}^\infty
	e^{(\hat{Z}^2-s^2)/2} \sinh\vovn s \left( \cosh\vovn s - \cosh\vovn \hat{Z}\right)^{k-1} \mathrm{d}s,
\end{align*}
where we apply the change of variable, $s = \vovn b$, to obtain the second line, and integration by part to obtain the third line.

Using the following analytic integral
$$ 
\int_z^\infty e^{-s^2/2} \sinh a s \, \mathrm{d}s = \frac{N(-z+a) - N(-z-a)}{2 n(a)} = \frac{N(z+a) - N(z-a)}{2 n(a)},
$$
for $k\neq 0$, we obtain
$$ \int_{\hat{Z}}^\infty e^{(\hat{Z}^2-s^2)/2} \sinh k\vovn s \, \mathrm{d}s = k\vovn\,m_k(\hat{Z})
\qtext{where}
m_k(\hat{Z}) := \frac{N(\hat{Z}+k\vovn)-N(\hat{Z}-k\vovn)}{2k\vovn\; n\left(\sqrt{\hat{Z}^2+(k\vovn)^2}\right)},
$$
which will be useful in the derivation below. 
Here, we define $m_k(\hat{Z})$ in such a way that $m_k(\hat{Z}) \to 1$ for any non-zero $k$ as $\vovn \to 0$.

Next, we evaluate the moments $E((\IVstep)^k)$, $k=1,2,3,4$. For $k=1$, we obtain
$$
\mu = \frac{e^{\vovn \hat{Z}}}{\vovn} \int_{\hat{Z}}^\infty
e^{(\hat{Z}^2-s^2)/2} \sinh\vovn s \, \mathrm{d}s 
= e^{\vovn \hat{Z}}\, m_1(\hat{Z}) = \left(\frac{\sigma_{t+\dt}}{\sigma_t}\right)\, m_1.
$$

For $k=2$, using $\sinh x\cosh x = \frac12 \sinh 2x$, we have
\begin{align*}
	\mu'_2 &= \frac{e^{2\vovn \hat{Z}}}{\vovn^3} \int_z^\infty
	e^{(\hat{Z}^2-s^2)/2} \sinh\vovn s \left( \cosh\vovn s - \cosh\vovn z \right) \, \mathrm{d}s \\
	&= \frac{e^{2\vovn \hat{Z}}}{\vovn^3} \int_z^\infty
	e^{(\hat{Z}^2-s^2)/2} \left( \frac{\sinh2\vovn s}{2} - \sinh\vovn s\cosh\vovn z \right) \, \mathrm{d}s \\
	& = \left(\frac{\sigma_{t+\dt}}{\sigma_t}\right)^2 \frac{1}{\vovn^2} \left(m_2 - c m_1 \right).
\end{align*}

For $k=3$, using $ \sinh x\cosh^2 x = \frac14 \left(\sinh 3x + \sinh x\right) $, we obtain
\begin{align*}
	\mu'_3 &= \frac{e^{3\vovn \hat{Z}}}{2\,\vovn^5} \int_z^\infty
	e^{(\hat{Z}^2-s^2)/2} \sinh\vovn s \left( \cosh\vovn s - \cosh\vovn z \right)^2 \, \mathrm{d}s \\
	&= \frac{e^{3\vovn \hat{Z}}}{2\,\vovn^5} \int_z^\infty
	e^{(\hat{Z}^2-s^2)/2} \left( \frac{\sinh 3\vovn s+\sinh\vovn s}{4} - \sinh2\vovn s\cosh\vovn z + \sinh\vovn s \cosh^2\vovn z \right) \, \mathrm{d}s \\
	& = \left(\frac{\sigma_{t+\dt}}{\sigma_t}\right)^3 \frac{1}{8\vovn^4}\left[3 m_3 - 8c m_2 + \left(4c^2 + 1\right)m_1 \right].
\end{align*}

For $k=4$, using $\sinh x\cosh^3x = \frac18 \left(\sinh 4x + 2\sinh 2x\right)$, we have
\begin{align*}
	\mu'_4 &= \frac{e^{4\vovn \hat{Z}}}{6\,\vovn^7} \int_z^\infty
	e^{(\hat{Z}^2-s^2)/2} \sinh\vovn s \left(\cosh\vovn s - \cosh\vovn z \right)^3 \, \mathrm{d}s \\
	&= \frac{e^{4\vovn \hat{Z}}}{6\,\vovn^7} \int_z^\infty
e^{(\hat{Z}^2-s^2)/2} \left[\frac{\sinh4\vovn s + 2\sinh2\vovn s}{8} - \frac{3}{4}\left(\sinh3\vovn s+\sinh\vovn s\right)\cosh\vovn z \right. \\
    & \hspace{12em} \left. +\,\frac32\sinh2\vovn s\cosh^2\vovn z - \sinh\vovn s\cosh^3\vovn z \right] \, \mathrm{d}s \\
	&= \left(\frac{\sigma_{t+\dt}}{\sigma_t}\right)^4 \frac{1}{24\,\vovn^6} \left[2 m_4 - 9c m_3 + \left(12c^2 + 2\right) m_2 - c\left(4c^2 + 3\right)m_1\right].
\end{align*}

In addition, the expansion of the mean, coefficient of variance, skewness, and ex-kurtosis (see Remark~\ref{r:mom}) around $\vovn=0$ are obtained as
\begin{gather*}
	Cv = \frac{\vovn}{\sqrt3} + O(\vovn^3), \quad
	s = \frac{6\sqrt3}{5} \vovn + O(\vovn^3), \qtext{and}
	\kappa = \frac{276}{35} \vovn^2 + O(\vovn^4).
\end{gather*}

\section{Sampling Islah's approximation via a CEV distribution} \label{apdx:islah_cev} \noindent
Though Islah's approximation of $F_{t+h}$ is less than desirable since it fails the martingale property, nevertheless, we show that Islah's approximation can be expressed as a power function of a CEV random variable. For theoretical interest, we show that the CEV simulation algorithm in Section~\ref{ss:cev_sample} can be applied to expedite the sampling of $F_{t+h}$ based on Islah'a approximation.

We define
\begin{gather*}
	\beta' := \frac{\beta}{1 - \betac\rho^2}\;(\ge \beta), \quad \betac' := 1-\beta' = \frac{\betac\rhoc^2}{1 - \betac\rho^2} \; (\le \betac) ,\\ 
	y' := \left(\frac{\betac'}{\betac} y^{\betac}\right)^{1/\betac'}, \qtext{and}
	\condFstep['] := \left| \frac{\betac'}{\betac} \left(F_t^{\betac} + \frac{\betac\rho}{\vov}(\sigma_{t+\dt} - \sigma_t)\right)\right|^{1/\betac'}
\end{gather*}
so that 
$$ \frac{(y')^{2\betac'}}{(\betac')^2} = \frac{y^{2\betac}}{\betac^2}
\qtext{and} \frac{\left(\condFstep[']\right)^{2\betac'}}{\betac'^2} = \frac{1}{\betac^2}\left[F_t^{\betac} + \frac{\betac\rho}{\vov}(\sigma_{t+\dt} - \sigma_t)\right]^2.
$$
As a result, the complementary CDF of Islah's approximation is equivalently expressed by
\begin{equation*}
	\mathrm{Prob}(F_{t+\dt} \ge y) 
	= \mathrm{Prob}\left(\left(\frac{\betac'}{\betac} F_{t+\dt}^{\betac}\right)^{1/\betac'} \ge y'\right)
	= P_\NCX\left(\frac{\left(\condFstep[']\right)^{2\betac'}}{\betac'^2\rhoc^2\sigma_t^2 h \IVstep}; \frac{1}{\betac'}, \frac{\left(y'\right)^{2\betac'}}{\betac'^2\rhoc^2\sigma_t^2 h \IVstep}\right),
\end{equation*}
and this implies that
\begin{equation} \label{eq:CEV_islah}
	\left(\frac{\betac'}{\betac} F_{t+\dt}^{\betac}\right)^{1/\betac'} \sim \cev[']\left(\condFstep['], \rhoc^2\sigma_t^2 \dt \IVstep\right).
\end{equation} 
In other words, instead of $F_{t+\dt}$ itself, a power function of $F_{t+\dt}$ follows a CEV distribution under Islah's approximation. Therefore, $F_{t+\dt}$ can be sampled as
$$ F_{t+\dt} \sim \left( \frac{\betac}{\betac'} \right)^{1 /\betac} Y^{\betac'/\betac} \qtext{where} Y \sim \cev[']\left(\condFstep['], \rhoc^2\sigma_t^2 \dt \IVstep\right).$$
Lastly, we obtain
$$ E\left(F_{t+\dt}^{\betac/\betac'} \,|\, \sigma_{t+\dt}, \IVstep\right) = \left( \frac{\betac}{\betac'} \right)^{1 /\betac'} E(Y) =  \left( \frac{\betac}{\betac'} \right)^{1 /\betac'} \condFstep['] = \left|F_t^{\betac} + \frac{\betac\rho}{\vov}(\sigma_{t+\dt} - \sigma_t)\right|^{1/\betac'}.$$

\begin{singlespace}
\bibliography{SABR}
\end{singlespace}
\end{document}